\documentclass[LaM,binding=0.6cm]{sapthesis}
\usepackage{cite}
\usepackage{relsize}
\usepackage{amsmath,amssymb,amsfonts}
\usepackage{graphicx}
\usepackage{textcomp}
\usepackage[svgnames]{xcolor}
\def\BibTeX{{\rm B\kern-.05em{\sc i\kern-.025em b}\kern-.08em
    T\kern-.1667em\lower.7ex\hbox{E}\kern-.125emX}}
\usepackage[ruled,vlined]{algorithm2e}
\usepackage[noend]{algpseudocode}
\usepackage{epigraph}
\usepackage{listings}
\usepackage{adjustbox}
\usepackage{multirow}
\usepackage{array}
\usepackage{float}
\usepackage[hyphens]{url}
\usepackage{microtype}
\usepackage{subcaption}
\usepackage{makecell}
\usepackage{pifont}
\usepackage{framed}
\usepackage{afterpage}
\usepackage{hyperref}
\hypersetup{pdftitle={Program State Abstraction for Feedback-Driven Fuzz Testing using Likely Invariants},pdfauthor={Andrea Fioraldi},hidelinks}

\usepackage{lipsum}
\usepackage{curve2e}
\definecolor{gray}{gray}{0.4}

\definecolor{shadecolor}{named}{LightGrey}

\newenvironment{example}%
  {\noindent\raggedright 
      \begin{shaded}\textbf{Example. }}{ 
   \end{shaded}}

\usepackage[colorinlistoftodos]{todonotes}

\usepackage{amsthm}
\newtheorem{definition}{Definition}
\newtheorem{theorem}{Theorem}

\newcommand{\aflpp}{{\sc AFL\nolinebreak[4]\hspace{-.05em}\raisebox{.45ex}{\relsize{-2}{\textbf{++}}}}\xspace}

\newcommand{\afl}{{\sc {AFL}}\xspace}

\newcommand{\redqueen}{{\sc {RedQueen}}\xspace}

\newcommand{\weizz}{{\sc {Weizz}}\xspace}

\newcommand{\aflsmart}{{\sc {AFLSmart}}\xspace}

\newcommand{\libfuzzer}{{\sc {LibFuzzer}}\xspace}

\newcommand{\nautilus}{{\sc {Nautilus}}\xspace}

\newcommand{\kafl}{{\sc {KAFL}}\xspace}
\newcommand{\mopt}{{\sc {MOpt}}\xspace}
\newcommand{\qemu}{{\sc {QEMU}}\xspace}

\newcommand{\llvm}{{\sc {LLVM}}\xspace}
\newcommand{\clang}{{\sc {Clang}}\xspace}
\newcommand{\valgrind}{{\sc {Valgrind}}\xspace}
\newcommand{\gcc}{{\sc {GCC}}\xspace}

\newcommand{\daikon}{{\sc {Daikon}}\xspace}
\newcommand{\invscov}{{\sc {InvsCov}}\xspace}
\newcommand{\sage}{{\sc {SAGE}}\xspace}
\newcommand{\perffuzz}{{\sc {PerfFuzz}}\xspace}
\newcommand{\fuzzfactory}{{\sc {FuzzFactory}}\xspace}
\newcommand{\syzkaller}{{\sc {SyzKaller}}\xspace}

\title{Program State Abstraction for Feedback-Driven Fuzz Testing using Likely Invariants}
\backtitle{Program State Abstraction for Feedback-Driven Fuzz Testing \\ using Likely Invariants}

\author{Andrea Fioraldi}
\IDnumber{1692419}
\course[override]{Degree course in Engineering in Computer Science}
\courseorganizer{Faculty of Information Engineering, Computer Science and Statistics}
\AcademicYear{2019/2020}
\copyyear{2020}
\advisor{Dr. Daniele Cono D'Elia}
\coadvisor{Prof. Davide Balzarotti}
\authoremail{andreafioraldi@gmail.com}

\begin{document}

\frontmatter

\maketitle

%\dedication{Dedicated to Kenneth McCormick}
%\dedication{To my grandmother L., who when she was a child, could not continue studying because despite the darkest hours of our country had just passed, the hearts of many were still black.}
\dedication{To my grandmother L., \\ who when she was a brilliant child,\\ despite the darkest hours of our country had just passed,\\ could not continue studying because the hearts of many were still black.}

\begin{abstract}
Fuzz testing proved its great effectiveness in finding software bugs in the latest years, however, there are still open challenges. Coverage-guided fuzzers suffer from the fact that covering a program point does not ensure the trigger of a fault. Other more sensitive techniques that in theory should cope with this problem, such as the coverage of the memory values, easily lead to path explosion. In this thesis, we propose a new feedback for Feedback-driven Fuzz testing that combines code coverage with the ``shape'' of the data. We learn likely invariants for each basic block in order to divide into regions the space described by the variables used in the block. The goal is to distinguish in the feedback when a block is executed with values that fall in different regions of the space. This better approximates the program state coverage and, on some targets, improves the ability of the fuzzer in finding faults. We developed a prototype using \llvm and \aflpp called \invscov.
\end{abstract}

\begin{acknowledgments}
I want to thank my advisors Davide and Daniele for the work done, in temporal order, to support me in the journey of the master thesis, despite a global pandemic and a busy student (myself) with too many other projects to carry on. I want to thank my family, my old and new friends, and the hackers' community that always I'm proud to be part of for the support and encouragement to pursue my work. This thesis was made mostly during my internship in the S3 lab of EURECOM, that I want to thank for the funding, the burned CPU cores, and the nice welcome.
\end{acknowledgments}

\tableofcontents

%\afterpage{\blankpage}
\mainmatter

\chapter{Introduction}

In the last two decades {\em Fuzz Testing} (or {\em Fuzzing}) gained popularity thanks to its ability in finding software bugs more effectively than other Software Testing techniques.

It is employed every day since 2016 in Google's {\sc OSS-Fuzz}~\cite{oss-fuzz} to continuously discover vulnerabilities in open source software and thousands of them were discovered so far in this four years of activity of the program.

This popularity comes also with the attention of academia and the industry on improving fuzzing techniques.

In recent times, on top of the twist of fuzz testing called {\em Coverage-guided Fuzz Testing (CGF)}, several new techniques were developed trying to overcome the limitations of CGF like hard to bypass path constraints~\cite{qsym}~\cite{sebastian}~\cite{driller}~\cite{vuzzer}~\cite{redqueen} or the high number of invalid testcases generated by mutation~\cite{aflsmart}~\cite{zest}~\cite{nautilus}~\cite{grimoire}~\cite{weizz}.

These improvements, with some exceptions like~\cite{parmesan}, try to improve the ability of a fuzzer to reach more code coverage. Code coverage is used as a proxy to {\em Program State Coverage} to avoid path explosion because the number of program states can be potentially infinite.
{\em Symbolic Execution}~\cite{baldoni} for instance does not employ such approximation, and in fact, path explosion is one of the most critical problems that make pure symbolic-based approaches impractical in real-world targets.

There is empirical evidence that fuzzers that uncover more code coverage discover also more bugs in programs. This can be motivated by the observation that exploring a portion of code is a necessary condition to find a fault in that code portion.

However, this is not a sufficient condition.

Fuzzers can saturate in coverage and never reach the combination of program states that leads to a bug. To cope with this problem, a fuzzer should observe the progress also in the program state data, not only in the control flow.

Recent works like~\cite{fuzzfactory}~\cite{ijon}~\cite{becollab} try to go beyond simple code coverage as feedback; we refer to these techniques --- CGF included --- as {\em Feedback-driven Fuzz Testing}.

Some fuzzers approximate the program state using more sensitive feedbacks, like code coverage with call stack information or even code coverage and values loaded and stored from memory. This second approach, as shown by~\cite{becollab}, better approximates the program state coverage by taking into account not only control flow but also the values in the program state data, but it is less efficient in finding bugs because of path explosion.

At the time of writing, the only successful approximation of the program state coverage using also values from the program state's data is done surgically on targeted program points selected by a human~\cite{ijon}. Portions of the state space are manually annotated and the feedback function is modified to explore such space more thoroughly.

The automation of this process is a crucial topic in future research in this field. 

In this thesis, we propose a new feedback for Fuzz Testing that takes into account not only Code Coverage, but also some interesting portions of the program states in a fully automated manner and without incurring path explosion.

We augment classic {\em Edge Coverage} --- a type of code coverage based on edges in the {\em Control Flow Graph} --- with information about ``unusual'' values in the program states observed in the incoming basic block.

We learn constraints between variables in basic blocks from executions traces of an input corpus, generally a corpus that is the output of a previous CGF fuzzing run, that describe how variables are related to each other and that holds for all the executions observed so far.

These constraints are mined {\em Program State Invariants} over basic blocks.

Execution-based invariants mining techniques, like the one that we use based on~\cite{daikon}, suffer from the well-known {\em Coverage Problem} that means that the learned constraints may be only local properties of the observed corpus, and the violations of a learned invariant may not lead to a violation of the program specification. This, however, is not a problem for our purpose.

Local properties, if enough generic like the learned invariants tries to be, are still an interesting abstraction of the program state.

%\daniele{check second part of this sentence}
So we define a new feedback function that distinguishes the same edge with learned invariants in the incoming basic block that holds from the same edge with one or more learned invariant violated.

We develop a set of heuristic to produce invariants and techniques to effectively instrument programs with a low-performance overhead --- a very important metric in fuzzing --- and we implement them into a prototype called \invscov on top of \aflpp~\cite{aflplusplus}.

In our evaluation, we show that a feedback that takes into account the program state abstraction can uncover more or different, software bugs than CGF.

\section{Contributions}

The key contributions of this thesis are:

\begin{itemize}
\item A new feedback that uses an abstraction of the program state from mined invariants;
\item A prototype implementation based on \llvm and \aflpp called \invscov;
\item A systematization of the concepts behind Feedback-driven Fuzz Testing.
\end{itemize}

We plan to share the prototype as Free and Open Source Software.

\section{Structure of the Thesis}

In Chapter \ref{cap:testing} we describe the basis of Software Testing and introduce various key concepts including the invariants.
In Chapter \ref{cap:fuzzing} we generically describe Fuzz Testing and introduce a new abstract taxonomy for Feedback-driven Fuzz Testing. We discuss also some of the challenges of Feedback-driven Fuzz Testing. 
In Chapter \ref{cap:meth} we introduce our methodology to abstract the program state's coverage. In Chapter \ref{cap:impl} we present our prototype \invscov and the technologies on which it is based.
In Chapter \ref{cap:eval} we evaluate the prototype in terms of efficiency and effectiveness.
The thesis ends with the conclusions and the discussion of future directions in Chapter \ref{cap:conclusion}.

\chapter{Basics of Software Testing}
\label{cap:testing}

{\em Software Testing} is the process that analyzes a {\em System Under Test} (SUT) to detect differences between existing and required conditions and to evaluate its features~\cite{testing_standard}.
The most common embodiment of Software Testing is the process of finding software bugs.

Bugs cause the software to produce incorrect results or to behave unexpectedly.
Bugs are a serious matter, they affect the everyday life of every person that depends on modern technology and, in the worst cases, cause even huge losses in terms of money~\cite{ariane} and human lives~\cite{therac}.

\section{Correctness}

%\daniele{shape the existence is a bit weird, like you create them?}
%\andrea{Si, sono astrazioni, osservi gli stage e le "crei". Il termine non è mio btw.}
In every stage of the system development, we can shape the existence of two entities: a {\em specification} and an {\em implementation}~\cite{laycock1993theory}.

The development process converts the specification into the implementation. A high-quality implementation means to satisfy as much as possible the specification.

An implementation that completely matches the specification would provide the highest quality, but such equivalence is impossible to be stated for a Software Testing process, otherwise, with such a process it would also be possible to solve the Halting Problem~\cite{rice}.

So the {\em correctness} of a SUT cannot be expressed in terms of equivalence between specification and implementation.

In the following, we provide some needed definitions that match with~\cite{glossary_standard}.

\begin{definition}
A {\tt failure} is an externally visible deviation from the specification.
\end{definition}

\begin{definition}
A {\tt fault} (or {\tt bug}) is portion of system state that leads to a failure. Note that if such a state exists but is never reached it does not cause a failure.
\end{definition}

\begin{definition}
An {\tt error} is a human error that causes the system to behave as not expected. Errors can cause faults.
\end{definition}

Given these concepts, we define the correctness as follows.

\begin{definition}
A {\tt correct} implementation of a specification does not contain faults.
\end{definition}

Correctness can be achieved by constructing a system without errors or by detecting and fixing all the faults.

In the first case, the absence of errors has to be proved formally. 

% todo explain that are both impossible in practice

\section{Validation and Verification}

The quality of a system in Software Testing is assessed by two processes: {\em Validation} and {\em Verification}~\cite{softw_book}.

Validation is the process that evaluates if the system really meets the needs for whom it was built. It is a subjective process that includes for instance user evaluations and prototyping. We will not discuss validation in this thesis.

Verification is the process, more objective than Validation, that evaluates if the implementation behaves according to the specification.

%\daniele{value inteso come "valore, senso" giusto?}
%\andrea{Si, hello world lo verifichi facile che non ha errori, ma serve a nulla.}
Verification processes aim to remove all the faults from the system, but this does not guarantee of course that the system has a value, this is stated by Validation.

Note that sometimes, for instance in the beta testing stages of the software development cycle, the two processes are combined.

\section{Properties of Testing}

% link to verification

For the sake of the verification of a system, a testing procedure may continue to add tests until all the faults are uncovered.

However, with constrained resources, this is impossible, so tests have to be prioritized.

Another property of testing is the context-dependent nature of the techniques.
Testing the autopilot software of a plane requires far different methods than testing a toy like a Tamagochi.

Faults are often clustered, bugs are not uniformly distributed in a system. This, for instance, affects the prioritization of the tests.

Another very important property is, using Dijkstra's words, that software testing can be used to show the presence of bugs, but never to show their absence.

So, given these properties, we want a great diversity in our tests and so it is convenient to automatically generate tests.

\section{Automation in Testing}

Automation in Software Testing, most of the time, means sampling the input space of a SUT to generate testcases.
The goal, according to~\cite{testing_efficiency}, is to gain confidence about a certain degree of correctness or to find as many faults as possible.

There are also, however, unsuccessful techniques that try to find just one faulty input to prove incorrectness. The failure in proving incorrectness, of course, does not prove correctness, so we exclude these techniques from the treatment.

There are two types of sampling:

\begin{enumerate}
  \item {\em Systematic}, in which the generation is informed by some artifacts from the SUT, like the specification;
  \item {\em Random}, a uniform at random sampling of the input space that basically has no cost;
\end{enumerate}

Given that we excluded techniques that find just a failing input, we can reduce Systematic sampling to a partitions-based approach.

Each partition is a subdomain of the input space and the inputs in each partition have common properties.

An effective type of partition strategy, as shown by~\cite{error_part}, is the one that samples from {\em error-based} partitions.
Each partition triggers an error or not. The testing strategy, given that is not known if a partition is associated with an error, samples each partition using a systematic approach.
This is very effective but also hard to apply in the real world. Rather than that, several other automatic techniques are used in practice and we discuss some of them in Sec \ref{sec:testingtech}.

\subsection{Efficiency Criteria}

As the goal of Automated Testing is to gain a certain level of confidence about the correctness of the SUT or to uncover as much as bugs as possible, an efficiency criteria, as introduced in~\cite{testing_efficiency}, must relate these goals with time.

In particular, in the first case the goal of an Automated Testing is to establish the level of confidence in a minimal amount of time, and, in the second case, to maximize the number of found bugs in a given time.

When evaluating Systematic techniques, an useful insight from~\cite{testing_efficiency} is that if we increase the effectiveness of the technique we have to increase the cost and so decrease the efficiency.
This leads to a second important result that is when the SUT size is over a certain bound, Random Testing becomes more efficient than the evaluated Systematic approach.

\section{Testing Techniques}\label{sec:testingtech}

We can divide testing techniques into different families. In this section, we will discuss some of them.

\subsection{Specification-based Testing}

To automatically generate a test several sources of information can be used. The first source of information for a Systematic approach is the specification.

Techniques based on this concept are called {\em Specification-based Testing} and do not require any knowledge of the structure of the program (i.e. programming language, size of the codebase, etc.).

The key idea is that several tests are derived from the specification and each test covers a {\em partition} of the System Under Test.

A common criteria, in term of partition-based testing, is to divide the input space into unique partitions that are unique in terms of exercised program behavior and in which is easy for an oracle to verify if such behavior, given one input, is correct.

\subsection{Structural Testing}
\label{sec:structural}

Another type of technique that aims to generate testcases using the code itself as a source of information is called {\em Structural Testing}.

A piece of typical information extracted directly from the code is the notion of {\em Coverage}, a notion that indicates how much code is exercised executing an input.

A prominent objective of this type of testing is to test all the code maximizing the coverage seen, but other criteria are supported as well and this affects how tests are generated.

Several types of coverage can be defined, following we discuss the most common classes:

\begin{itemize}
\item {\em Line Coverage} is probably the most straightforward type of coverage, related to lines fo code covered. However, this is a problematic coverage because it is affected by the coding style and programming language density;

\item {\em Block Coverage} dope with the limitation of line coverage defining a more objective metric. Block refers to a block in the {\em Control Flow Graph (CFG)}~\cite{cfg}, which is defined at a high level as the paths that can be traversed in the code. A {\em Basic Block} is an aggregation of adjacent code lines executed without a control flow change (i.e. no branches), a {\em Decision Block} is a block containing the predicate that affects a control flow change and an {\em Edge} is the connector between these blocks. A basic block has only one exit, a decision one has two exits, one of the condition is true, the other if it is false. Block coverage is simply when the testing technique aims to cover all the blocks in the CFG instead of all the lines in the code;

\item {\em Edge Coverage} is always related to the CFG and it is used when block coverage is not enough in presence of complex branch conditions. Covering the edge coverage at 100\% means that all the decision branch are exercised;

\item {\em Path Coverage} is the coverage of all possible independent paths in the code. In term of the CFG, maximizing the path coverage means covering each possible path from each node to each other connected node;

\end{itemize}

\subsection{Model-based Testing}

A model of a system under test holds some properties and attributes of such a system in an abstract way.
Using such abstraction tests can be generated.

A widely used type of model is the {\em Decision Tables}, tables that relate actions that the systems perform and conditions to take that action.

A prominent version of model-based testing is {\em State-machine based Testing}.

A state machine describes the system using states and transitions between these states.

To derive testcases from a state machine, like for structural testing, we can define some types of test coverage:

\begin{itemize}
  \item State coverage: all the states have to be covered at least once;
  \item Transition coverage: each transition has to be covered at least once;
  \item Path coverage: exercise combination of transitions called paths;
\end{itemize}

The typical model-based testing workflow is bringing the system into different states and, after each transition, asserting that the system is in the expected new state.

\subsection{Property-based Testing}

In structural testing, we use information from the code to generate tests. In {\em Property-based Testing}, we use properties of the program to let the code itself to check the correctness.

The oracles that check if an execution of a test is related to a correct behavior are not anymore external but embedded in the code.

A common construct that developers employ to do that is the notion of {\em assertion}, a boolean expression inserted in a specific program point that, if false, reveals the presence of a bug.

Testcases can be then generated until an input that violates one of the assertions is found. {\em QuickCheck}~\cite{quickcheck} is one of the first tools developed that generates almost random inputs for the program under test in order to find testcases that violate one or more assertions.

Related to this type of testing is the software design pattern called {\em Design-by-contracts}. In this methodology, each caller component ensures that the preconditions to call a callee component are met.
The main advantage is that errors are caught by the code itself avoiding propagation and the computation of incorrect results.

The checks that test the code during the execution are based on the concept of {\em Invariant}, a property that is always true at one or more particular program points, as described in~\cite{ernst_phd}.

\subsubsection{Types of Invariants}

Two widely known types of invariants are pre and post conditions of programs. Firstly introduced by~\cite{Hoare1969}, with the term {\em Hoare Triples} we denote the triple $\{P\} A \{Q\}$, where P is the pre-conditions that holds before the execution of the program A, and Q are the post-conditions that holds after the execution. Of course, A can denote also functions in a program or even single statements.

\begin{example}
Consider the following C function that pop an element from a stack data structure:

\begin{lstlisting}[language=C]
struct item* stack_pop(struct stack* stack) {
  struct item* item = stack->base;
  stack->base = item->next;
  stack->size--;
  return item;
}
\end{lstlisting}

In this case, we can define this Hoare Triple:

\begin{itemize}
\item $P$: \lstinline[columns=fixed, language=C]{stack != NULL, stack->base != NULL, stack->size > 0};
\item $A$: the \lstinline[columns=fixed, language=C]{stack_pop} function;
\item $Q$: \lstinline[columns=fixed, language=C]{stack->size == original(stack->size) -1}, \lstinline[columns=fixed]{return == original(stack->base)};
\end{itemize}

You can easily see that if one of the preconditions is violated there is a fault (a Null Pointer Dereference or an Integer Overflow).

\end{example}

Another type of invariants are {\em Class Invariants}~\cite{hoare72}, that are strictly related to Object-oriented Programming.
These invariants refer to an object or class and hold for the entire lifetime of such an object.

Loop invariants~\cite{floyd1993assigning}~\cite{Hoare1969} are predicates over the state of a loop that holds for every loop execution.

Also, other types of invariants were defined in the literature, for instance like the invariants based on concurrency constraints in~\cite{lamport1977proving}. 

\subsubsection{Mining Invariants}
\label{subsec:mining}

Automatic invariants learning is a widely used process in verification, for instance for testcases generation~\cite{csallner2005check} or memory errors detection~\cite{reed1991purify}.

A valid approach to learn invariants from the code of a SUT is static analysis, using techniques like {\em Symbolic Execution}~\cite{baldoni}, of which~\cite{tillmann2006discovering} is an example, or {\em Abstract Interpretation}~\cite{blanchet2002introduction}, used in~\cite{giuffrida_psi}.

While invariants extracted using this type of analysis most of the time correct and without false positives, they are often overapproximation or simply the static analysis is not enough powerful to spot some invariants that are revealed only at runtime.

Opposite to that, many approaches like~\cite{diduce}~\cite{ernst2001dynamically}~\cite{pattabiraman2010automated} uses information gathered during the execution, in a dynamic fashion.

One of the most popular tools to extract invariants from program traces is \daikon~\cite{daikon}, that uses a machine learning approach to learn from traces that are in an abstract form, allowing the tool to support many programming languages and runtimes like C, Java, C\#, and many more.

The problem of these approaches is, unlike static approaches, that the extracted constraints are {\em likely invariants}, properties that hold at least for the observed executions. This, while allowing the analysis to reveal more invariants, leads to false positives with constraints that are only local properties of the observed executions. This problem is commonly called the {\em Coverage Problem} since, if the corpus of testcases used for learning does not cover the possible program state, at it is like to be in real applications, there may be still testcases with different coverage that violates the learned likely invariants.

%\section{Discussion}

\newpage

% efficiency vs effectiveness

% https://mboehme.github.io/paper/TSE15.pdf

\chapter{The Art of Fuzzing}
\label{cap:fuzzing}

In this chapter, we systematically review the concepts behind {\em Fuzz Testing}, a successful Random Testing family of techniques.

In particular, we focus on {\em Feedback-driven Fuzz Testing}, a technique that evaluates generated testcases using feedback from the {\em System Under Test (SUT)}.

\section{Generic Definitions}
\label{sec:gendef}

% \daniele{PUT not defined before. Define also PPUH. And explain the subtle difference between PUT and SUT}
\begin{definition}
{\tt Fuzz testing} or {\tt Fuzzing}, according to {\rm~\cite{survey}}, is the repeated execution of the Program Under Test (PUT) using inputs sampled from an input space and that stresses the PUT with unexpected inputs. We can generalize Fuzzing using SUT as subject instead of PUT.
\end{definition}

The usage of SUT is motivated by the fact that Fuzz Testing is nowadays extended to other domains, for instance, it can be used to test a web application composed of different distributed components or a set of programs communicating with each other like the IPC stack of a browser.

We simply call a program that implements a Fuzzing technique {\em Fuzzer}.

\begin{definition}
A {\tt Violations Oracle} is a process that determines if an execution of the SUT violates some requirements.
\end{definition}

An example of requirements is the correctness of the SUT or some performance requirements. We can also, for instance, assert that some code must not be reached if a specific configuration of the SUT is provided and use Fuzzing to try to find an input that executes that code under that configuration. Typically, many fuzzers look for crashes in a program, an easily observable type of failure.

%\daniele{Define artifacts}
\begin{definition}
The {\tt Fuzzer State} is the set of variables and artifacts (e.g. the instrumentation added to the SUT if any) that affects the behavior of the Fuzzer. They can evolve during the fuzzing process.
\end{definition}

Note that the fuzzer state is called configuration in~\cite{survey}, however, this is misleading while talking about tunable fuzzers like~\cite{aflplusplus}, so we use the term State instead.

A generic enough algorithm that can be used to define fuzz testing is \ref{alg:fuzz}.

\vspace{2mm}
\begin{algorithm}[H]
\label{alg:fuzz}
\DontPrintSemicolon
\KwResult{The set V of testcases with violations}
 $V \gets \emptyset$\;
 $S \gets \textsf{Preprocess}(S)$\;
 \While{$\textsf{Continue}(S)$}{
  $I \gets \textsf{InputGeneration}(S)$\;
  $S, V \gets \textsf{InputEvaluation}(I, O, S, V)$\;
 }
 \Return{$V$}\;
 \caption{Generic Fuzz testing}
\end{algorithm}
\vspace{2mm}

The abstract stages in the algorithms are defined accordingly to~\cite{survey} but in a more generic flavor.

We define such stages as follows.

\begin{definition}
The {\tt Preprocess} stage is executed once before the fuzzing loop and modify the initial Fuzzer state.
\end{definition}

\begin{definition}
The {\tt Continue} stage decides, looking at the current state, if the fuzzer must stop or continue to search for violations.
\end{definition}

\begin{definition}
The {\tt InputGeneration} stage generate a testcase for the current run using the state.
\end{definition}

Some Fuzzers can use for instance previous testcases embedded in the state to generate a new testcase.

\begin{definition}
The {\tt InputEvaluation} stage is the stage responsible to feed the SUT with the newly generated input and use the violations oracle $O$ to determine if this is a testcase that violates some of the requirements. It also updates the Fuzzer state with information useful in other stages.
\end{definition}

\section{Fuzzers Classification}

Until now we defined fuzzing in a generic way, unrelated to the actual System Under Test or to the specification.

We can go further in the definition of Fuzz Testing using classes based on how the Fuzzer generates the testcase and how much information about the SUT it needs.

Firstly, looking at the InputGeneration stage, we can define two types of generations:

\begin{itemize}
  \item {\em Model-based} generation uses a model of the input format of the SUT embedded in the Fuzzer State to generate the testcase from scratch. It can be for instance a grammar specified in the initial Fuzzer state by a human (e.g.~\cite{langfuzz}), a mined model using learning techniques (e.g.~\cite{learnfuzz}) or hardcoded generation rules in the algorithm (e.g.~\cite{csmith});
  \item {\em Mutation-based} generation uses previous testcases called {\em corpus} that has to be provided in the initial state or saved in the state during the previous executions of InputEvaluation to generate the testcase modifying a testcase in the corpus;
\end{itemize}

Note that the mutational generation can use different strategies to mutate an input, even a model of the input format like model-based generation like in~\cite{nautilus}~\cite{aflsmart}. The difference, in this case, is that the new testcase is not generated from scratch.

Orthogonally to the properties of InputGeneration, we can classify a fuzzer using the information that it needs from the actual SUT.

There are three common classes of fuzzers based on this criteria~\cite{godefroid2007random}:

\begin{itemize}
  \item {\em Black-box} fuzzers does not need any insight from the SUT. The first attempts at fuzzing are of this type, which is closely related (if not even equal) to traditional Random Testing. Note that the lack of information from the actual implementation does not imply the lack of information about parts of the specification. For instance, black-box fuzzers like Peach~\cite{peach} require a model of the input format to generate testcases;
  \item {\em White-box} fuzzers systematically inspect the state space of the SUT using internals information. White-box fuzzing can often overlap with Systematic Structural Testing. An example is \sage~\cite{sage} that tries to maximize code coverage using constraints gathered during the execution;
  \item {\em Grey-box} fuzzers stands in the middle of the two previous approaches. They collect minimal information from the SUT to better explore the input space while maintaining the performance overhead low. Traditionally, the information is collected during the SUT execution and it is code coverage, like in~\cite{aflwhitepaper};
\end{itemize}

The distinction between these categories is often unclear, they are commonly used in the Security community but it is a debatable taxonomy, in this traction, we will avoid it and use the taxonomy from Software Testing when possible.

\section{Feedback-driven Fuzzing}

There is a strong empirical evidence~\cite{oss-fuzz} that fuzzing with code coverage as feedback from the SUT increases the efficiency in terms of the number of founds faults in a given time window.

Many fuzzers rely on code coverage (mostly use edge coverage~\cite{aflwhitepaper}~\cite{libfuzzer}) as a feedback, but the technique is not limited to this particular information and, as shown by~\cite{ijon}~\cite{fuzzfactory}~\cite{becollab}, there can be many types of feedbacks (Sec. \ref{sec:feedback}) that a fuzzer can use to better explore the input space.

Fuzzers typically employ evolutionary algorithms to process the collected information from the SUT. A general enough version of those algorithms is Alg. \ref{alg:evolutionary}.

\vspace{2mm}
\begin{algorithm}[h]
\DontPrintSemicolon
\KwResult{The set V of testcases with violations}
 $V \gets \emptyset$\;
 $S[Corpus] \gets InitialCorpus$\;
 $S[SUT] \gets \textsf{Instrument}(S[SUT])$ \algorithmiccomment{Preprocess}\;
 \While{$\textsf{Continue}(S)$}{
  $T \gets \textsf{PickTestcase}(S)$\;
  $N \gets \textsf{Calibrate}(T, S)$\;
  \For{$i \gets 0$ \textbf{ to } $N$}{
    $I \gets \textsf{InputMutation}(T)$ \algorithmiccomment{Mutation-based InputGeneration}\;
    $S, V \gets \textsf{InputEvaluation}(I, O, S, V)$\;
  }
 }
 \Return{$V$}\;
 \caption{Basic Evolutionary Fuzz testing}
 \label{alg:evolutionary}
\end{algorithm}
\vspace{2mm}

\iffalse
\begin{algorithm}[H]
\label{alg:feedback}
\DontPrintSemicolon
\KwResult{The set V of testcases with violations}
 $V \gets \emptyset$\;
 $S[Corpus] \gets InitialCorpus$\;
 $S[SUT] \gets \textsf{Instrument}(S[SUT])$\;
 $OBS \gets \textsf{OpenObservationChannels}(S[SUT])$ \algorithmiccomment{Preprocess}\;
 $FBS \gets \textsf{SetupFeedbacks}(OBS, S)$\;
 \While{$\textsf{\upshape Continue}(S)$}{
  $T \gets PickTestcase(S)$\;
  $N \gets Calibrate(T, S)$\;
  \For{$i \gets 0$ \textbf{ to } $N$}{
    $I \gets InputMutation(T)$ \algorithmiccomment{Mutation-based InputGeneration}\;
    $S, V \gets InputEvaluation(I, O, S, V)$\;
  }
 }
 \Return{$V$}\;
 \caption{Feedback-driven Fuzz testing}
\end{algorithm}
\fi

InputEvaluation in Alg. \ref{alg:evolutionary} is responsible to evolve the Corpus in the fuzzer state.

In Feedback-driven Fuzzing it executes the SUT with the generated input and, looking at the gathered feedback, if the execution is interesting it adds the testcase to the corpus. We can formalize Feedback-driven InputEvaluation as in Alg. \ref{alg:input}.

\vspace{2mm}
\begin{algorithm}[h]
\DontPrintSemicolon
\KwData{The generated input I, the violations oracle O, the current fuzzer state S and the set V of testcases with violations}
\KwResult{The next fuzzer state S and the updated set V of testcases with violations}
 $Tr \gets \textsf{Execute}(I)$ \algorithmiccomment{Tr is the observed trace of the execution}\;
 \If{$\textsf{IsViolation}(O, Tr)$} {
  $V \gets \textsf{AddToSet}(V, I)$\;
 }
 \If{$\textsf{IsInteresting}(S, Tr)$} {
  $S \gets \textsf{EvolveCorpus}(S, I)$\;
 }
 \Return{$S, V$}\;
 \caption{InputEvaluation in Feedback-driven Fuzz testing}
 \label{alg:input}
\end{algorithm}
\vspace{2mm}

IsInteresting is deeply connected with the feedback chosen, which has to maintain in the fuzzer state progress information. In the simple case of code coverage, in the state, there is the coverage seen so far and IsInteresting returns true only if there is a previously unseen coverage in $Tr$.

To better define Feedback-driven Fuzz Testing, we can systematically define all of its entities and how they are related to the abstract stages defined in \ref{sec:gendef}. Note that these concepts are also abstract themself and we can use them to classify feedback-driven fuzzers.

% executor etc

\subsection{Oracle}
\label{sec:orac}

As defined in \ref{sec:gendef}, a {\em Violation Oracle} inspects the execution of the SUT to decide if a testcase violates or not the expected requirements.

For Feedback-driven Fuzz Testing we add another property to the oracles, concerning the feedback: the ability to distinguish between violating inputs that are duplicates in terms of violated requirements or some other criteria.

\begin{definition}
An {\tt Oracle} is an entity that inspects the execution of the SUT to determine if it violates the given requirements. Besides, it decides if a violating testcase is worth to be added to the set of violating testcase V based on some criteria.
\end{definition}

We said that this new property is somewhat related to the feedback because, typically, fuzzers uses the feedback to distinguish between crashing inputs that probably trigger the same bug. For instance, \afl~\cite{aflwhitepaper} use the coverage to distinguish crashes. If a new crash does not trigger new coverage in relation to the testcases already in V, it is discarded as a duplicate.

Other fuzzers like~\cite{honggfuzz} use the hash of the callstack just after the crash.

Another common type of oracle is the one employed by differential fuzzing that tests if the outcome of an implementation is the same of the outcome of another considered correct~\cite{aumasson2017automated}.

\subsection{Observation Channel}
\label{sec:obs}

% The feedback from the SUT is gathered using an {\em Observation Channel}.

\begin{definition}
An {\tt Observation Channel} is an entity that provides information about the SUT execution to the fuzzer.
\end{definition}

The execution trace $Tr$ in \ref{alg:input} is a possible outcome of an {\em Observation Channel}. It is used to get the information needed by the feedback.

An example of an observation channel is the shared memory of \afl. It is a shared bitmap between the target and the fuzzer that reports the coverage.

Note that the feedbacks uses the observation channels, but they are not limited to be used only for feedbacks~\cite{weizz}.

\subsection{Executor}

The concept of executing the SUT is not always the same. For instance, for in-memory fuzzers like \libfuzzer~\cite{libfuzzer} the SUT is a harness function, for hypervisor-based fuzzers like \kafl~\cite{kafl} instead the SUT can be the entire operating system.

All fuzzers, however, needs the same primitives to execute a SUT.

\begin{definition}
An {\tt Executor} is an entity with a set of violation oracles \ref{sec:orac}, a set of observation channels \ref{sec:obs}, a function that allows instructing the SUT about the input to test, and a function to run the SUT.
\end{definition}

Placing an input in the SUT can be a very different task depending on the type of executor. \libfuzzer just place it as arguments of the harness function, while other fuzzers can write to the program standard input or in a specific memory region when executing the target inside an emulator or a similar controlled environment.

\subsection{Feedback}
\label{sec:feedback}

% \daniele{Maybe the first time you mention feedback in early chapters you can add a forward pointer to this subsection}
As the technique is Feedback-driven, the concept of {\em Feedback} is quite important.

The fuzzer must interpret the information retrieved using the observation channels and relate it to the fuzzer state.

\begin{definition}
A {\tt Feedback} is an entity that defines how to interpret the gathered information from observation channels and how to evolve the corresponding fuzzer state, specifically how to evolve the testcases corpus based on that information.
\end{definition}

For instance, given an observation channel that reports the size of memory allocations, a feedback that aims to maximize these allocations size to spot out-of-memory bugs can be defined with the following sentence:

Given a map $M$ from the observation channel, each entry corresponds to a single program point and it is a 64-bit unsigned integer. The fuzzer state must maintain an accumulation map $A$ that, for each registered allocation program point, maintains the maximum size of the allocation seen so far. A testcase is added to the corpus (IsInteresting) if when updating $A$ with the corresponding $M$, at least one entry is updated.

In literature, some fuzzers make use of feedbacks that does not simply aim to maximize code coverage, like \perffuzz~\cite{perffuzz} that maximizes the execution counts of all the program locations to spot performance issues and \fuzzfactory~\cite{fuzzfactory} that implements different feedback functions based on maps.

Note that many fuzzers use a map as observation channel and a reduce function to evaluate if the collected information is interesting, but the notion of feedback is not limited to this particular embodiment.

\subsection{Input}

Until this point, we did not specifically define what is an input for the SUT because it is an abstract concept. It is, in general, a sample from the {\em Input Space}, all the possible data that the SUT can take from an external source and that affects its behavior.

In the straightforward case, the input is a simple file or buffer passed to a program, but it can be also, for instance, a sequence of actions~\cite{aflnet} or even different values read in different program points independently~\cite{vfuzz}.

\begin{definition}
An {\tt Input} entity defines one possible sample from the Input Space and can hold properties about the input itself, the relation between the input and the SUT, or the input and the specification.
\end{definition}

Note that the description of how the SUT consumes the input is provided by the Executor.

\subsubsection{Input metadata}

We refer to the properties that the Input can hold as {\em Input metadata}.

An example of metadata, that relates the input to the specification, is the virtual structure.

Fuzzers such as \aflsmart~\cite{aflsmart} or \nautilus~\cite{nautilus} maintains a tree representing the Abstract Syntax Tree of the input when parsed using an {\em Input Format Specification} provided by the user. That information comes from the specification of the SUT (e.g. we know that the program under test process PNG files) and it is used in the mutator \ref{sec:mut} to perform structure-aware mutations.

Another example of metadata is the tags extracted by \weizz~\cite{weizz} from the SUT using dynamic analysis. Like in the previous example, this metadata is used for mutation, but it relates the input with the SUT and not with the specification, because they are extracted using an observation channel.

A third, and naive, example of metadata that related the input with the SUT is the execution time, and one that is a property of the input itself is the size in bytes if it is a buffer. This kind of metadata is often used in \afl-like fuzzers to prioritize simpler inputs.

\subsubsection{Scheduling inputs}

The outcome of Calibrate in evolutionary fuzzing is affected by the current Input $T$ extracted from the corpus.

Calibrate controls how many fuzzing iterations have to be done mutating a testcase. It can employ algorithms based on the input metadata, for instance, the execution time.

These algorithms that schedule how many iterations are assigned to a testcase are commonly referred to as {\em Power schedules}. Several works addressed this problem like~\cite{aflfast}~\cite{entropic}, using, for instance, the triggered coverage as input metadata. An interesting insight from these works is that is convenient for some fuzzers to give more iterations to testcases that cover program points that are rarely stressed.

\subsection{Corpus}

The Corpus is a set of inputs in the fuzzer state. In Feedback-driven Fuzzing, it is necessarily connected to the Feedback.

\begin{definition}
A {\tt Corpus} is an entity that collects inputs that are interesting for one or more feedbacks, and defines how they are related to each other and how to feed the fuzzer with those inputs when requested.
\end{definition}

The Corpus is implemented as a data structure containing inputs, but it can differ a lot for each fuzzer.

For instance, \afl uses a queue to store inputs, but other fuzzers use just a simple container that feeds the fuzzer using a random selection algorithm.

The Corpus is also not unique. There can be a corpus that contains inputs interesting for just one feedback.

\subsubsection{Corpus transitions}

The Corpus can use inputs metadata to relate the inputs to each other and schedule how to serve testcases to the fuzzer when it needs them for a fuzzing run.

We can view the Corpus as an evolving entity not only when adding a new testcase, but also when the fuzzer requests the next testcase to fuzz.

We call this second type of evolution from a request to another {\em Corpus transition}.

When a Corpus randomly selects the next testcase, the transition is the simplest. When it is a queue, the transition is the \verb|get| operation of the queue.

More complex fuzzers employ algorithms to select smartly the next testcase. For instance, \afl uses the coverage triggered by each input to create a minimized subset of the Corpus that covers the entire coverage seen so far. With a high probability, only inputs in this subset, that is periodically updated, are fuzzed.

% favoreds set al.

\subsection{Mutator}
\label{sec:mut}

\begin{definition}
A {\tt Mutator} is an entity that takes one or more inputs and generates a new derived one.
\end{definition}

A mutator can both modify the input and the related metadata. Some mutators work just on the metadata and then the change is replaced to the testcase.

The concept of mutator is deeply linked with the definition of the input for the SUT, typically for each type of input, there are specialized mutators.

\subsubsection{Scheduling mutations}

A mutator most times can apply more than a single type of mutation on the input. Consider a generic mutator for a byte stream, bit flip is one of the possible mutations but not the single one, there is also, for instance, the random replacement of a byte of the copy of a chunk.

When a mutator has a collection of mutations, what mutations have to be used for the specific input is a problem addressed scheduling the mutations.

Naively, for most fuzzers, the number of mutations is a random bounded number and the sequence of mutations is randomly chosen too. In more complex approaches, a scheduling algorithm is chosen.

For instance, \mopt uses particle swarm optimization to select mutations based on their effectiveness in finding new interesting input in past iterations of the fuzzer.

\subsection{Generator}

% talk here about structured fuzzing in relation with testcase metadata

\begin{definition}
A {\tt Generator} is an entity that generates a new input from scratch possibly using some parameters.
\end{definition}

Opposed to input creation by mutation, there is input creation by generation. Feedback-driven Fuzzing is most of the time related to Evolutionary Fuzz Testing that, of course, needs a mutator to evolve the corpus. There are, however, situations in which generators are used in this kind of testing.

One is when the generator is invoked by a mutator, like when mutating a virtual structure of the input. Consider a mutator that operates on the Abstract Syntax Tree, it can randomly replace a subtree with a new one generated from scratch as a mutation. In this case, generation is a mutation.

Another situation is when the SUT itself asks for the input the fuzzer in different stages (e.g. a stateful network protocol when a request-response sequence is considered a single input), and the entire input that was initially created by mutation was already given to the SUT. In this case, the fuzzer has to generate an extension of the input previously generated by mutation and can do that using a generator.

A third situation, less explored, is using Feedback-driven Fuzzing not to evolve a corpus of testcases, but a corpus of parameters for the generator. Consider a grammar-based generator, a parameter that affects the generations is, for instance, the probability to choose a terminal node or go deeper in the grammar and continue exploring a nonterminal. These parameters can potentially be tuned using the feedback.

\subsection{Stage}

\begin{definition}
A {\tt Stage} is an entity that operates some actions on a single input.
\end{definition}

This definition of Stage is very abstract, we used it because it is just a proxy used as a building block of the fuzzing algorithm.

A mutational stage, given an input of the corpus, applies a mutator and executes the generated input one or more time. How many times this has to be done can be scheduled, \afl for instance use a performance score of the input to choose how many times the havoc mutator should be invoked. This can depends also on other parameters, for instance, the length of the input if we want to just apply a sequential bitflip, or be a fixed value.

A stage can be also an analysis stage, for instance, the {\em colorization} stage of \redqueen that aims to introduce more entropy in a testcase or the {\em trimming} stage of \afl that aims to reduce the size of a testcase. A possible stage can also, for instance, execute the SUT with particular instrumentation to extract some input metadata, like in \weizz.

\section{Challenges}

% roadblocks, invalid inputs, saturation, path explosion, bug trigger (asan)

The latest state-of-the-art research in Fuzz Testing tries to address some of the challenges that make fuzzing less efficient.

\subsection{Roadblocks}

The most straightforward limitation, when dealing with coverage as feedback, is the code roadblocks for the fuzzer. Multi-byte comparisons are one type of roadblock because, given a generic byte stream mutator, it is nearly impossible to guess the exact value of the other operand of the comparison to flip the branch.

\begin{example}

Consider the following C code snippet:

\begin{lstlisting}[language=C]
void foo(int x) {
  if (x == 0xbabdcafe)
    bug();
}
\end{lstlisting}

If x is the input provided by the fuzzer, the probability that the branch is flipped guessing that x should be 0xbabdcafe is almost 0.
\end{example}

In literature, this issue is addressed using a feedback that track the progress of the comparison~\cite{laf}~\cite{value-profile}, with concolic execution combined with fuzzing~\cite{qsym}~\cite{sebastian} or with techniques that extract the comparison values and try to replace patterns in the input~\cite{vuzzer}~\cite{redqueen}.

Other common roadblocks are the checksums. They are mostly used to protect chunks of binary formats against corruption, the probability that a generic fuzzer mutates the protected bytes and, at the same time, restore the checksum field validity is almost 0.

This problem can be addressed using a specific mutator for the binary format or with code transformation that patch the program removing the checksum checks~\cite{redqueen}~\cite{tfuzz}~\cite{weizz}.

\subsection{Invalid inputs}

Another challenge of fuzzers with generic mutators is the high rate of generated invalid inputs.

When the mutator likely breaks the validity of the input, the fuzzer stresses most of the times the code related to parsing, but not deeper code.

In order to effectively fuzz deep paths, the fuzzer must produce valid inputs, and this can be achieved using a model of the input format to guide the mutator, like in~\cite{aflsmart}~\cite{nautilus}.

Some techniques tried to approximate that automatically, without a human written model~\cite{grimoire}~\cite{weizz}~\cite{zeller_issta20}.

Another approach is to constrain the mutator to not touch the portion of the input that leads to the deep path and even constraint the possible values that an input field can take using, for instance, constraints collected using concolic tracing~\cite{pangolin}.

\subsection{Faults without Failures}

Most oracles in fuzzers use failures to know if a testcase violates the requirements, in particular crashes, but sometimes a fault does not directly trigger a failure. 

To catch these kinds of bugs, the SUT is often instrumented with additional tripwires to catch silent faults. For instance, a one-byte overflow in read on the heap will unlikely trigger a crash in a C program. To handle this situation, source-based fuzzers offer the possibility to instrument the programs with sanitizers such as {\sc AddressSanitizer}~\cite{asan}.
Others make use of binary-only tripwires to uncover silent corruptions~\cite{marius}, inserted both dinamically, like the \qemu-based sanitization available for \aflpp~\cite{qasan}, or statically~\cite{retrowrite}~\cite{revngfuzz}.

These sanitizers however cannot catch some pure logic bugs, and fuzzing to uncover this kind of bugs automatically (e.g. without putting manually assertions in the code) is an open field of research.

\subsection{State Tracking}

The code coverage, or some extension of it like comparisons feedback, is often not enough to explore the state space of a SUT. For instance, the code that handles a specific type of chunk in a binary format is considered always the same in terms of coverage, even when the previously processed chunk is different. The interleaving of different chunks is still an interesting property of those programs and often bugs are related to it, but most fuzzers cannot get feedback from it because code coverage does not notice it.

The current method to overcome this challenge is to manually select some state variables and get feedback from them manually~\cite{ijon}.

The technique proposed in this thesis tries to address this particular challenge automatically.

\subsection{Path Explosion}

When the sensitivity of a feedback increases, for instance when adding feedback about the progress of comparisons to edge coverage or when using context-sensitive coverage, the fuzzer may save as interesting too many testcases. These testcases will be never processed all, and the fuzzer saturates. An example of feedback that easily leads to path explosion, as described in~\cite{becollab}, is memory coverage, a feedback that considers a testcase as interesting if the execution of the SUT uses a previously unseen zone of memory.

Path coverage too easily leads to path explosion, and it is one of the most prominent problems of techniques that use it like Symbolic Execution~\cite{baldoni}.

\subsection{Scaling Implementations}

A current limitation of available fuzzers implementation is scaling on multiple CPU cores. System calls used by the fuzzer for various tasks are often not designed for scaling, they use expensive locks in the kernel or performs unneeded slow tasks for our purpose. The usage of Inter-Process Communications primitives offered by the OS to communicate between the fuzzer and the target instrumented program is an example.

A possible solution to this problem is embedding the fuzzer in the target application itself, like \libfuzzer does, and use a custom kernel extension or hypervisor to snapshot the program state in case of code that cannot be fuzzed stateless. Mocking the syscalls performed by the target helps too, usually, this task is performed when using an emulator to instrument the program.

Besides that, note that an exponential increase of the cores, as shown in~\cite{empiricalLaw}, increases only linearly the ability to find new coverage or new bugs.

\subsection{Hard Targets}

The problem of instrumenting, executing, and feeding a SUT with the produced testcases is not trivial in practice.

While traditional tools like \afl instruments programs that take their input as a file, many applications and systems are not designed to behave in that way.

On the other hand, many programs can be adapted to consume a buffer produced by the fuzzer, like when using a \libfuzzer harness, and others can be instrumented to request pieces of input when needed to the fuzzer instead of using a single buffer~\cite{vfuzz}, but there are targets that consume their inputs in other ways that are either not easy to handle or for which even execution or instrumentation is hard.
%that cannot be even easily executed or instrumented.

In the first case, a SUT of this kind is a complex system like a kernel or a multiprocess environment like IPC stacks. For instance, \syzkaller~\cite{syzkaller} provides inputs to the kernel using a sequence of syscall invocations and fuzzing using a hypervisor with incremental snapshots~\cite{whatthefuzz} can be used to test multiprocess environments that share messages.

In the second case, the SUT is typically an embedded system. Often the execution of firmware is possible only on specific hardware, which makes instrumentation and hardening to catch silent faults impossible~\cite{marius}. To address this problem, the code can be executed on an emulated hardware~\cite{firmafl}~\cite{usbfuzz}, with a high cost in terms of development effort, or re-hosted~\cite{halucinator}~\cite{unicorefuzz}, a technique that transfers the context from the target to an emulator and forwards the interactions with the hardware back to the device only when needed.

% \section{Implementations}

% parla tipo dei fuzzer tipo aflfast in termini delle parti astratte definite prima
% e di come addressano i problemi

\section{Evaluation Criteria}

The evaluation of fuzzing techniques is a matter currently under debate in the community. 

A recent work~\cite{fuzzeval} analyzed some papers on fuzzing and stated what should and should not be done to compare fuzzers. The number of testcases or reported crashing inputs is not an evaluation metric, and using them is a very bad practice. This work states that an effective metric is the number of triaged bugs triggered by each fuzzer or the coverage over time as a good proxy.

However, this approach is limited and too general when we want to evaluate the properties of fuzzers. A fuzzer that triggers less but unique bugs is good and, in with the same spirit, a fuzzer that covers unique coverage is good too. Recently, this metric of counting unique code blocks covered was introduced in FuzzBench~\cite{fuzzbench}.

Other metrics that are useful to evaluate fuzzers are which fuzzer reach a certain coverage in less time~\cite{gamozo} and the number of hits of each block for some randomly sampled generated input, in order to evaluate the ability of a mutator to stress deep paths~\cite{blog_cornelius}.

%\newpage

%\daniele{Perhaps you could be adding a table with some of the most notable fuzzers and their characteristics wrt some of the main categories mentioned early. Also, you should be adding a paragraph where you build the momentum for your proposal, motivating your work as it tackles a compelling issue etc}

\newpage

\chapter{Methodology}
\label{cap:meth}

Code coverage as feedback for Feedback-driven Fuzz Testing is a successful proxy to approximate the program state during the exploration of the paths performed by the fuzzer.

Edge Coverage-based Fuzzing found thousand and thousand of bugs in complex applications in the lastest years~\cite{oss-fuzz}, however, it suffers in exploring program states that lead to bugs but that are not directly related to code coverage.

A naive solution may be a Fuzzing algorithm that uses code and variables or memory values as feedback, but this quickly leads to path explosion, as shown by~\cite{becollab}.

In this chapter, we describe or technique that tries to cope with this problem augmenting the feedback given by classical Edge Coverage-based Fuzzing using {\em basic blocks invariants} violations to approximate the program state coverage.

\section{Definitions}
\label{sec:methdefs}

In this section, we introduce some needed definitions for the rest of the chapter, in part according to~\cite{daniele_pldi18}.

\begin{definition}
A {\tt Program} is a sequence of instructions.
\end{definition}

In this dissertation, the SUT is a program.

\begin{definition}
The {\tt Memory State} is a function $M(a) \rightarrow v$ that associates a {\tt Memory Address} $a$ to a value.
\end{definition}

\begin{definition}
The {\tt Static Single Assignment form (SSA)} {\rm~\cite{ssa}} is a type of intermediate representation (IR) in compiler theory. SSA requires that every variable is defined before each use and assigned only once.
\end{definition}

Without loss of generality, every program can be converted to SSA~\cite{ssa_convert}.

\begin{example}
Consider this simple C function that computes the maximum of two integers:

\begin{lstlisting}[language=C]
int max(int x, int y) {

  int m;
  if (x > y)
    m = x;
  else
    m = y;
  return m;

}
\end{lstlisting}

We can convert this simple program to SSA avoiding assigning two times the variable \verb|m|. At first glance, this seems impossible because its value depends on the control flow, but we can use the $\Phi$ operator of SSA. This operator defines a new variable choosing between two possible values depending on the control flow.

For our simple example, we can define three different SSA variables that represent \verb|m| in the two blocks of the if statement and in the terminal block that ends the function.

The translation to an SSA pseudocode is then the following:

{
\SetAlgoNoLine%
\setlength{\interspacetitleruled}{0pt}%
\setlength{\algotitleheightrule}{0pt}%
\setlength{\algoheightrule}{0pt}%
\begin{algorithm}[H]
\DontPrintSemicolon
\SetKwFunction{FMain}{max}
\SetKwProg{Fn}{Function}{}{}
\Fn{\FMain{$x$, $y$}}{
  \eIf{$x > y$}{
    $m_1 \gets x$\;
  }{
    $m_2 \gets y$\;
  }
  $m_3 \gets \Phi(m_1, m_2)$\;
  \KwRet $m_3$\;
}
\end{algorithm}
}

\end{example}

\begin{definition}
The {\tt Load} instruction $v := M(a)$ assign to a variable the value associated with the address $a$, that is a variable itself.
\end{definition}

\begin{definition}
The {\tt Store} instruction $M' := Store(M, a, v)$ alter the codomain of a memory state changing the value corresponding to $a$ with $v$.
\end{definition}

\begin{definition}
A {\tt Basic Block} {\rm~\cite{gcc_bb}} is a straight-line sequence of instructions in the program with only one entry point and only one exit.
\end{definition}

\begin{definition}
The {\tt Control Flow} is the order in which each individual basic block is executed and evaluated. A {\tt Control Flow instruction} is an instruction that changes that order.
\end{definition}

% \daniele{Fix the definition. A (previously defined) variable is live at a block if the block uses it or any block reachable from it may read from it later in the execution (well, before a variable redefinition intervenes, but with SSA this is not possible => you can explain this after the definition)}
\begin{definition}
The set of the {\tt Live Variables} at a basic block is the set of variables that are used in the block or possibly used in any block reachable from it later in the execution.
\end{definition}

This set can be calculated using a backward analysis called {\em Liveness Analysis}~\cite{daniele_pldi18}. Note that as in SSA a variable is used only once, redefinition is not an issue for our definition of liveness.

%\daniele{say tuple/triple, not set}
\begin{definition}
A {\tt Program State} is the triple $(l, V, M)$, in which $l$ is the program point of the next instruction to be executed, $V$ is the set of the live variables and $M$ is the current memory state.
\end{definition}

In our simplified model of the computation in a program, this definition of program state alone is enough to describe the data in the program.

Note that, if we have to observe a program state during the execution of a program, we have to do it after the execution of the current instruction, because the next instruction $l$ is determined only after the evaluation of the instruction in case of control transfer instructions.

\section{The Basic Block State}
\label{sec:bbstate}

\begin{definition}
The {\tt State of a basic block} is the set of all the SSA variables values used in the block.
\end{definition}

%\daniele{fix verbi parte finale}
Such states can be easily observed after the execution of the basic block because SSA variables are assigned only once and so all the used values are still in the variables.

We also define a family of functions that describe basic blocks.

%\daniele{I would say: the set of live SSA variables $V_I$ and the memory state $M_I$ as they are before the execution of X}
\begin{definition}
Given a basic block $X$, the function $BB_X(V_I, M_I) \rightarrow (V_O, M_O)$ is the function that takes the set of live SSA variables $V_I$ and the memory state $M_I$ as they are before the execution of X. $BB_X$ returns a new set $V_O$ of SSA variables that are the variables used in X (note that $V_I \cap V_O$ is, in general, not empty) and the memory state $M_O$ that is the memory state of the program after the execution of X.
\end{definition}

Following this notation, $V_O$ is thus the basic block state.

Always using the properties of SSA, we can derive this property of the basic block state:

\begin{theorem}
\label{teo:bbmem}
The side effect in memory of a basic block, $M_O \setminus M_I$, is a subset of the basic block state $V_O$.
\end{theorem}

\begin{proof}
In an SSA IR, only the store instruction can affect the memory state. A store instruction $s$ takes in input the memory location and the value $v$ that has to be stored and so if $s$ is part of a basic block, $v$ is a variable in the basic block state.
\end{proof}

\section{Program State Abstraction}

% TODO do not use an already built trace

%\daniele{maybe you need to define a composition operator for basic block functions; for the paper for sure}
\begin{definition}
Given an observed execution of a program, a {\tt Program Trace} can be described as a chain of $BB_X$ functions that take as input the output of the previous function in the chain. 
\end{definition}

The order of the functions in the chain is given by the order of the executions of the basic blocks in the observed execution of the program, which is the Control Flow.

We can use the program trace to observe the transitions between program states, and, using the following theorem, we can do it efficiently at the end of each block instead of for each instruction.

\begin{theorem}
\label{teo:trace}
Given a program trace, the observation of each basic block state from the outcome of each $BB_X$ function provides the same information about changes in the program state of observing the values during each program state transition.
\end{theorem}

%\daniele{una cosa forse ambigua nella definizione/proof: se una variabile smette di essere live ti cambia il program state, ma il cambiamento di liveness non lo puoi definire da qui. In generale il concetto di liveness dipende da tutti i path futuri su un CFG, qua sei su un trace ed i path futuri noti sai quali saranno (non necessariamente li prendi tutti). Tienine conto per il paper}
\begin{proof}
A transition in the program state causes a change in:
  \begin{enumerate}
  \item obviously, the program point;
  \item in the set of live variables, if a new variable is defined or another is not live anymore;
  \item in the memory state, if the current instruction is a store;
  \end{enumerate}

We know, from Theorem \ref{teo:bbmem}, that the new values in the memory store are also tracked in the basic block state.

Given that the control flow of a basic block is unique, if we observe the basic block state after the execution of a block, we get that:
  \begin{enumerate}
  \item each instruction was necessarily executed, so we know all the program points in the transitions;
  \item all the live variables for each instruction are still in the basic block state because, by definition, the liveness as a basic block granularity, and the newly defined variables are tracked because we observe them at the end of the block, so after all the definitions;
  \item the changes in the memory states are in the basic block state;
  \end{enumerate}

So, observing the changes in terms of values in the program states inside a basic block at the end of the execution of such block is the same as doing it for each instruction.
\end{proof}

Note that for our theoretical model we are assuming a program without exceptional control flow, as usually these paths are not interesting to fuzz because related to errors and our treatment is simplified.

%\daniele{Claim corretto, dato che non hai exceptional control flow - in practice? assumi exception-free program?}
%\andrea{mezzo si, magari lo motivo dopo in implementation, la metodologia è questa e le eccezioni sono poco interessanti per il fuzzer, stai esplorando error paths, che vuoi peraltro evitare}
For each block, this information is a valid approximation of observing each individual program state.

Obviously, in a real-world implementation, logging a program state is not feasible even for a small number of executed instructions due to the occupied space. Logging the changes is a good trade-off and, as shown by Theorem \ref{teo:trace}, we can do it at the end of each block.

%\daniele{la proiezione del BB state escludendo la memoria?}
%\andrea{si, definito nella sezione precedente che il bb state sono le SSA var a fine bb}
Now consider the basic block state as a space on the SSA variables. We can divide such space into subspaces using relations between the SSA variables.

These relations define hyperplanes in such space, and so also subspaces are implicitly defined by these hyperplanes.

% \daniele{fixa il respect if we do not}
\begin{example}
Consider a basic block state with two variables, $x$ and $y$.
If we have, for instance, the relations $y > x - 8$ and $x < 100$ we have that the following four subspaces are defined: $y > x - 8 \land x < 100$, $y > x - 8 \land x >= 100$, $y <= x - 8 \land x < 100$, $y <= x - 8 \land x >= 100$.
If we violate one of the relations, we are in another subspace than when we do not violate them. If we violate both, there is another subspace.
\end{example}

%\daniele{state or states?}
A main idea is that if these relations describe coherently different partitions of possible values of the variables (e.g. there is a bug if $x \% 2 = 0$) we can use the produced subspaces as an abstraction of the program state in the basic block.

\section{Mining Subspaces}

% \daniele{meaningful? instead of significative}
The relations between variables in a basic block state that delimitate the subspaces have to be chosen carefully in order to have a meaningful division of state space.

As our goal is to find bugs, an effective definition of such relations can be as {\em Basic Block Invariants}.
This type of invariants, used in works like~\cite{davide_php}~\cite{giuffrida_psi}, are relations that theoretically always hold in a basic block and describe all the possible values of a variable.

When one or more of such relations are violated, we are in a subspace that is related to an incorrect program state reached in this basic block.

%\daniele{il pezzo finale that use the basic block states chiariscilo meglio}
As the SSA form guarantees us that the input variables of a block are not modified, we can avoid checking the pre-conditions of a Hoare Triple over the basic block because the same invariants are also present in the post-conditions, that use the basic block state. Note that in the basic block state the memory information is only related to the changes made in the current block, so implicit constraints between variables and memory values are missing.

This could be an effective technique to detect bugs, but mining these relations is hard.

As described in \ref{subsec:mining}, some approaches try to learn the invariants with static analysis techniques, others try to learn the invariants from many executions traces of the program under test.

%\daniele{where/such that the violations? fixa anche i plurali/singolari violation(s) e regards}
The first approach may miss invariants and the second may generate relations in which the violations are not related to a real bug, but is a local violation regards the learned data from the corpus of the execution traces, a likely invariant.

These over-approximated invariants are however interesting even if the violations maybe not be related to a bug. In the abstraction of the program state, a local property is still interesting and the definition of subspaces based on this kind of invariants is still a valid approximation.

Thus, we can exploit the coverage problem to learn relations about common states of the variables and define the subspaces as spaces in which the variables assume ``unusual'' values that may or may not be related to a bug.

\section{An Invariants-based Coverage}

Given the relations extracted using an execution-based invariants mining technique, we can define a new type of feedback for Fuzz Testing that is based on the abstract program state coverage defined by the subspaces of the basic block state.

The control flow information is, like in traditional Coverage-guided Fuzzing, approximated by the observed edges in the Control Flow Graph.
This information is augmented with the subspace in which the variables of the incoming block are located.

So, a fuzzer saves an input not only when a new edge previously unseen is executed, but also when the program explores a new subspace of the incoming basic block state for that edge.

%\daniele{per i riferimenti alle sezioni puoi usare \\Cref del package cleveref}
\begin{example}

Consider the basic block from the previous example in Sec. \ref{sec:bbstate}. There are four possible spaces defined by two invariants.

Assuming that the block can generate $N$ possible branches, it can produce $N * 4$ different feedback items for the fuzzer.

However, the observation of all the $N * 4$ cases is an extreme case in which all the likely invariants can be violated at the same time. In practice, some invariants may never be violated when another is violated too, or never violated at all.

\end{example}

To check if a particular state violates the learned invariants for a block, this technique has to emit these checks at the end of the blocks using code generation.

During the execution, the boolean information about the violation of a single invariant is then propagated to the next executed block to report the edge information augmented with the position of the observed values in the defined subspaces of the incoming basic block.

\section{Pruning Invariants}

Variables in the basic block state are, however, not always related to each other and this can produce useless likely invariants.

Besides, invariants that are impossible to violate, even in case of a fault, are not relevant for our technique.

To cope with these two problems, which pollute our coverage and increase the number of checks that must be generated --- so increasing complexity and decreasing the execution speed of the program under test --- we devise some optimizations for the invariants miner.

On the side of invariants checking, after the miner extracts the likely invariants, we can avoid checking for duplicate invariants to speed up the execution.

\subsection{Comparability Calculation}
\label{sec:comp}

\vspace{2mm}
\begin{algorithm}[h]
\DontPrintSemicolon
\KwData{The IR function F containing all the IR instructions in that function}
\KwResult{The function $C \colon Instructions \longrightarrow Comparability$ that relates an instruction I to a comparability id}
 $C \colon \textsf{GetAllValues}(F) \longrightarrow \{\epsilon\} $\;
 \For{$I$ \textbf{\upshape{in}} $\textsf{GetInstructions}(F)$}{
  \If{$\textsf{IsUnary}(I)$}{
    $C, Id \gets \textsf{MergeComparability}(C, Id, I, \textsf{GetOperand}(I, 1))$\;
  }
  \ElseIf{$\textsf{IsCast}(I)$}{
    $C, Id \gets \textsf{MergeComparability}(C, Id, I, \textsf{GetOperand}(I, 1))$\;
  }
  \ElseIf{$\textsf{IsBinary}(I)$}{
    $C, Id \gets \textsf{MergeComparability}(C, Id, I, \textsf{GetOperand}(I, 1))$\;
    $C, Id \gets \textsf{MergeComparability}(C, Id, I, \textsf{GetOperand}(I, 2))$\;
  }
  \ElseIf{$\textsf{IsGEP}(I)$}{
    $C, Id \gets \textsf{MergeComparability}(C, Id, I, \textsf{GetPointerOperand}(I))$\;
    $O_1 \gets \textsf{GetIndexOperand}(I, 1)$\;
    \For{$O$ \textbf{\upshape{in}} $\textsf{GetIndexOperands}(I) \setminus O_1$}{
      $C, Id \gets \textsf{MergeComparability}(C, Id, O_1, O)$\;
    }
  }
  \ElseIf{$\textsf{IsLoad}(I)$}{
    $C(I) \gets Id$\;
    $Id \gets Id +1$\;
  }
 }
 \Return{$C$}\;
 \caption{Comparability set computation}
 \label{alg:comp}
\end{algorithm}
\vspace{2mm}

\vspace{2mm}
\begin{algorithm}[h]
\DontPrintSemicolon
\KwData{The function $C \colon Instructions \longrightarrow Comparability \cup \{\epsilon\}$, the progressive counter $Id$, the IR values $V_1$ and $V_2$}
\KwResult{The function $C \colon Instructions \longrightarrow Comparability \cup \{\epsilon\}$ and the progressive counter $Id$}
  \If{$C(V_1) \neq \epsilon \land C(V_2) = \epsilon$}{
    $C(V_2) \gets C(V_1)$\;
  }
  \ElseIf {$C(V_1) = \epsilon \land C(V_2) \neq \epsilon$}{
    $C(V_1) \gets C(V_2)$\;
  }
  \ElseIf {$C(V_1) \neq \epsilon \land C(V_2) \neq \epsilon$}{
    $C(V_1) \gets Id$\;
    $C(V_2) \gets Id$\;
    $Id \gets Id +1$\;
  }
  \Else {
    \For {$V$ \textbf{\upshape{in}} $\textsf{Domain}(C)$}{
      \If{$C(V) = C(V_2)$}{
        $C(V) \gets C(V_1)$\;
      }
    }
  }
 \Return{$C, Id$}\;
 \caption{MergeComparability auxiliary algorithm}
 \label{alg:mergecomp}
\end{algorithm}
\vspace{2mm}

Firstly, we try to create independent sets of related variables in a IR function.

%\daniele{ogni istruzione ha un solo comp set possibile, corretto? per come hai definito la funzione. A farla precisa dovresti definirla come $\longrightarrow$ Comparability $\cup$ \{$\epsilon$\}}
Variables that are not directly related but are related to a common third variable are in the same set. These sets are called {\em Comparability} sets and can be encoded using a function $C \colon Instructions \longrightarrow Comparability \cup \{\epsilon\}$ that relates each instruction in the function to a comparability set, identified by an ID in this case, or to default comparability $\epsilon$, that represents the comparability with all the other variables.

Algorithms \ref{alg:comp} and \ref{alg:mergecomp} describe how we compute $C$.

%\daniele{vedi needed aggiunto alla fine}
Each variable is initially related to $\epsilon$, then, the list of the instructions is walked and if the current instruction is of a certain kind (e.g. a binary operator) each variable is marked as related to the other. For instance, if the instruction is an addition, the comparability set of the result is merged with the sets of the two operands. The first time that a variable is hit, a new comparability id is assigned instead of $\epsilon$. Variables that are used in instructions that are not in the cases shown in \ref{alg:comp} maintain the $\epsilon$ comparability, an over-approximation needed to not lose interesting invariants.

The kinds of instructions that are considered in Alg. \ref{alg:comp} are the unary instructions, operators with a single operand, binary instructions, with two operands, cast instructions, that convert a value to another with a different type, load instructions, and GEP instructions, that compute an address given a pointer and a set of indexes.

\subsection{Inviolable Invariants}

As our technique is based on violations of likely invariants, we want to avoid the generation of invariants that are always inviolable, also in the case of a fault.
Learning basic invariants from the actual code of the program leads to such types of invariants.

\begin{example}

Consider a \lstinline[columns=fixed, language=C]{unsigned int} variable in C. It will be mapped to many different IR variables, but all of them will never be negative.
An inviolable invariant is so that these variables are always greater than or equal to 0. This is a useless check for the invariants' coverage, it will be never violated.

\end{example}

{\em Value Range Analysis}~\cite{range_analysis} is the technique that, if used with a conservative approach, allows us to define bounds to integer variables that always hold. 

Consider the definition of a constraint variable $Y$ in one of the following ways:

\begin{itemize}
  \item $Y = [a, b]$
  \item $Y = \textsf{Merge}(X_1, X_2)$
  \item $Y = X_1 + X_2$
  \item $Y = X_1 * X_2$
  \item $Y = a*X + b$
  \item $Y = X \sqcap [a, b]$
\end{itemize}

The $\textsf{Merge}$ operator merges two variable names into one~\cite{merge_vsa}, $\sqcap$ is the range intersection.

These definitions can be easily extracted from a program in SSA form.

The Range Analysis objective, as explained in~\cite{integer_range}, is to solve a constraint system with variables in the same form of $Y$ and associate each variable to an integer range (with $+\infty$ and $-\infty$).

So, with this technique applied to SSA variables used to produce invariants, we can extract over-approximated ranges for each variable and exclude these invariants from the learning output, as they always hold.

\subsection{Deduplicate Invariants}

Basic blocks defined in terms of a sequence of instructions in a program are, like generic basic blocks defined in \ref{sec:structural}, nodes in the {\em Control Flow Graph (CFG)}.

We can define dominance relations between blocks in the CFG as in~\cite{dominator}.

\begin{definition}
A basic block A {\tt dominates} a basic block B if every path in the CFG from the root to B must go through A.
\end{definition}

Related to this definition, there are other definitions:

\begin{definition}
A basic block A {\tt strictly dominates} a basic block B if A dominates B and A $\neq$ B.
\end{definition}

\begin{definition}
A basic block A {\tt immediately dominates} a basic block B if A strictly dominates B but does not strictly dominate any other block that strictly dominates B too.
\end{definition}

The immediate dominator is unique and every block (except the root one) has one.

\begin{definition}
The {\tt Dominator Tree} is a tree defined with the basic blocks as nodes and the edges as the immediately dominates relations.
\end{definition}

As the immediately dominates relation is unique for the dominated block, the dominator tree is a tree in which a block can immediately dominate multiple blocks but can be immediately dominated only by one.

From a node, we can go backward in the dominator tree to find all the nodes that strictly dominate such nodes.

We can make use of such definitions to develop a method to prune duplicate invariants between blocks.

%\daniele{la transizione ad instrumentation/analysis code placement arriva un po' dal nulla, tecnicamente comunque tutto ok!}
After the likely invariants phase, the produced checks have to be placed in the code and the outcome of each check is the identifier of the invariant if violated.
If two IR variables are used in two blocks, the blocks will likely share an invariant.

With the dominator tree of a function, we can optimize this phase and avoid the insertion of redundant checks.

The actual check is emitted only for the top-level (the nearest to root) dominator that shares the invariant, and the outcome is propagated to the dominated blocks without the need to execute again the check.

\section{Corpus selection}
\label{sec:corpus}

We rely on a mining approach based on execution traces and this needs several testcases to generate different traces. Like for previous evolutionary fuzzing techniques, the choice of the initial corpus is critical.

%\daniele{avevo fixato come: "in which indeed the input corpus affects the ability of the fuzzer in exploring coverage a lot" ma poi non si capisce se vuoi dire che in un caso hai negative effect e l'altro no oppure se per entrambi hai negative effect. Rivedile se riesci queste due righe, si capisce bene solo quando arrivi a spiegare il secondo fuzzer}
Unlike traditional CGF, in which indeed the input corpus affects by a lot the performance of the fuzzer, an unwise choice of the initial corpus for our technique can produce biased results, not just degrade
 performance.

For instance, it is a common practice to download many files of a given file format when testing a parser, but those files are almost all valid files. If we learn likely invariants from the execution of a similar corpus, we will bias our invariants on the validity of the file format and, in some cases, this can be a mistake because we miss interesting partitions of the basic block states related to invalid inputs.

As our technique aims to improve the ability of a fuzzer to better explore the code regions that were already reached executing the initial corpus, an interesting choice is to mine invariants over the corpus of another fuzzer.

This is interesting because we have nowadays fuzzers that reach very good coverage in a reasonable time (e.g.~\cite{redqueen}~\cite{sebastian}).

A derived problem is then when should we stop the first fuzzer and apply our technique?

We can randomly select a time window, or wait that the fuzzer saturates in coverage~\cite{saturation}. 

The saturation is a known problem in fuzzing, and doing the latter can help to cope with this problem, but sometimes a fuzzer can get stuck in the opposite way: it will continue to find new coverage and fail to deeply explore a single code region. In this case, a random timeout is reasonable, maybe combined with the partial instrumentation of some selected regions of the programs to avoid this problem.

\section{Discussion}

With the proposed feedback in this chapter, an edge can be registered as different feedback values and introduce novelty more than a single time. To do that, we divide the space described by the IR values used in the incoming basic block using functions over these values. These functions must meaningfully represent interesting properties of such space, and so we used learned likely invariants over the basic block.

We devised a set of algorithms to reduce the number of generated invariants that are redundant, that do not produce feedback or that relate unrelated variables.

The proposed technique aims to augment the feedback to use not only information about the control flow but also about the entire program state to better explore the possible states of the SUT during Fuzz Testing.

\newpage

\chapter{Implementation}
\label{cap:impl}

In this chapter, we provide an overview of the technologies used to implement our technique, \llvm~\cite{llvm}, \daikon~\cite{daikon}, and \aflpp~\cite{aflplusplus}, the general architecture of our implementation and some details about it.

\section{The Low Level Virtual Machine Infrastructure}

The {\em Low Level Virtual Machine Infrastructure} (\llvm)~\cite{llvm} is a compiler and toolchain infrastructure designed for easy development of compiler frontends for programming languages and backends for instruction set architectures. 

The core of \llvm is its intermediate representation (IR) that is frontend-agnostic, portable, and SSA-compliant. This IR allows compiler architects to implement optimizations and code analysis passes at many stages of the compilation pipeline in a completely language-independent flavor.

%\daniele{arg => argument? magari usa una opt maggiormente famosa, tipo global value numbering}
The generation of the machine code in \llvm can happen at compile-time, for each module, at link-time, enabling a wide range of aggressive inter-procedural optimization like arguments promotion~\cite{LattnerAdve:tutorial}, or even at run-time, using the \llvm just-in-time engine.

The supported backends at the time of writing are almost all the most used architectures, including X86, PowerPC, ARM, and SPARC, but also less known ones like Hexagon or WebAssembly.

Since our techniques are designed to work at the IR level, we can build an implementation that is architecture-independent.

\llvm has many available frontends for many programming languages too, like C, Rust, Go, C++, and Ada, that emit IR for the backend after language-dependent optimizations. The basic data types in the IR are integers and floats, and there are five built-in derived types: pointers, arrays, vectors, structures, and functions. If the language targeted by the frontend supports more data types, they are expressed in the IR as a combination of the standard IR datatypes.

The first-class citizen frontend of \llvm is \clang, the C, C++, Objective-C, and Objective-C++ frontend.

We built our prototype for C/C++ programs, so we rely on \clang as frontend.

The structure of the intermediate representation is SSA, as written before, that makes use of an infinite set of registers. The IR is also strongly typed and RISC-like.

\begin{example}

An example of human-readable IR function is the following:

\begin{lstlisting}[language=llvm]
define i32 @max(i32 %x, i32 %y) {
entry:
  %cmp = icmp sgt i32 %x, %y
  br i1 %cmp, label %if.then, label %if.else

if.then:
  br label %if.end

if.else:
  br label %if.end

if.end:
  %m.0 = phi i32 [ %x, %if.then ], [ %y, %if.else ]
  ret i32 %m.0
}
\end{lstlisting}

It is the translation of the SSA pseudocode discussed in Sec. \ref{sec:methdefs}.

All the variables are assigned only once and the main difference with the SSA pseudocode is the missing definition of the intermediate \verb|m| variables. \llvm does not support definitions of variables directly using another one because you can just use the original variable (\verb|%x| and \verb|%y| in our case) instead.

\end{example}

%\daniele{in-tree non si capisce se non sei LLVM expert, rephrase}
The \llvm infrastructure is modular: you can write a so-called \llvm pass that can manipulate the IR and it is invoked during various stages of the IR optimization. Passes can be in-tree, modifying the source tree of \llvm itself like the AddressSanitizer pass, or out-of-tree, that are shared object loaded at runtime.

\llvm incorporates also a debugger, a C++ standard library implementation, and a linker to enable link-time optimizations.

\section{The Daikon invariant detector}

\daikon is a dynamic miner of likely invariants, previously introduced in \ref{subsec:mining}. 

It relies on execution traces that are in a language-independent form, enabling the tool to operate on data produced by different tracers, most notably the Java tracer and the C tracer based on \valgrind~\cite{valgrind}, and even fictional traces that are not outcomes of the execution of a program.

It requires a declaration (\verb|decls|) file that describes the traced variables and one or more data trace (\verb|dtrace|) files that report the actual values of the traced variables.

In both formats variables are grouped by the program point in which they are observed, typically function enter and exit.

The pattern for a program point entry in the declaration file is\footnote{\url{https://plse.cs.washington.edu/daikon/download/doc/developer/File-formats.html\#Program-point-declarations}}:

\begin{verbatim}
ppt <ppt-name>
<ppt-info>
<ppt-info>
...
variable <name-1>
  <variable-info>
  <variable-info>
  ...
variable <name-2>
  ...
\end{verbatim}

%\daniele{dtaflow? in verbatim or define it}
The \verb|ppt-name| encodes the name plus its type, for instance \verb|name:::ENTER| for variables when entering a function, \verb|name:::EXIT| for variables when exiting a function, or \verb|name:::OBJECT| for object fields.
\verb|ppt-info| encodes properties such as if a method is private and the parent program point in the \verb|dataflow hierarchy| (e.g. an object program point is the parent of all the method program points of such object because a method can access the fields).

%\daniele{fixa i plurali sono scozzati}
The variable names must be unique inside the program point, and the \verb|variable-info| holds information about the data type of the variable, the comparability (Sec. \ref{sec:comp}), the bounds if any, if it is a function parameter and other generic flags.

In a data trace, the corresponding program point entry has the following format:

\begin{verbatim}
<program-point-name>
this_invocation_nonce
<nonce-string>
<var-name-1>
<var-value-1>
<var-modified-1>
<var-name2>
<var-value-2>
<var-modified-2>
...
\end{verbatim}

The nonce is a progressive number to define a total order between such entries. The logged values for each variable are followed by the \verb|var-modified| field, which tells if the variable was modified since the last time it was logged.

\section{The AFL++ fuzzing framework}

{\em American Fuzzy Lop ++} (\aflpp)~\cite{aflplusplus} is a recent fork of the popular coverage-guided fuzzer \afl, that is not improved anymore with novel features from 2017.

\aflpp incorporates and sometimes reimplements some of the latest relevant research in Fuzz Testing, such as~\cite{mopt},~\cite{redqueen} and~\cite{aflfast}.

It supports many different instrumentation backends to extract coverage information from the target, for both compiler-based instrumentation and binary-only instrumentation.

To instrument a program during the compilation pipeline, \aflpp ships a \gcc plugin and a set of \llvm passes. For binary-only targets, \aflpp has two Dynamic Binary Translation backends that instrument the code during the recompilation stage of the JIT, one based on \qemu~\cite{qemu} and the other based on Unicorn Engine~\cite{unicornemu}.

It supports different types of code coverage, such as standard edge coverage, context-sensitive edge coverage, ngram coverage, and more. The standard edge coverage is logged in a shared map (\verb|__afl_area_ptr|) used as an observation channel. The number of hits for each edge is logged too.

The classic instrumentation, inherited from \afl, inserts a snippet similar to the following C code at each basic block:

\begin{lstlisting}[language=C]
void afl_maybe_log(unsigned cur_loc) {

  static __thread unsigned prev_loc = 0;
  __afl_area_ptr[cur_loc ^ prev_loc]++;
  prev_loc = cur_loc >> 1;

}
\end{lstlisting}

When using the compiler wrapper that uses the \llvm passes to instrument the code, the equivalent of this code is inserted inline at the IR level.

The \verb|cur_loc| parameter is generated at compilation time and identifies the basic block while \verb|prev_loc| maintains the information about the last executed block. So the shared map index \verb|cur_loc ^ prev_loc| is related to the currently executed edge.

\section{The \invscov pipeline}

We implemented our technique in a prototype called \invscov that stands for Invariants Coverage.

The \invscov pipeline is composed of several stages:

\begin{enumerate}
\item Dumper compilation: a first version of the program that dumps variables values is compiled;
\item Online learning: the dumper program is executed for each input in the corpus and \daikon performs online learning of the likely invariants;
\item Checks generation: the output of \daikon is processed and an object file with all the checks for each invariant is produced;
\item Target compilation: a second version of the program is compiled instrumenting the blocks with the classic \afl instrumentation plus the calls to the checks (generated C functions) needed for the invariants coverage;
\end{enumerate}

%During the entire pipeline, 

\subsection{Dumper compilation}

The dumper instrumentation is handled by an \llvm pass on functions and a runtime object.
This pass is the one that implements Algorithm \ref{alg:mergecomp} (Sec. \ref{sec:comp}) and also uses Range Analysis to learn the bounds of the integer IR values, when possible.

As an implementation choice, in this pass we reduced the number of IR values considered as variables for the invariants miner if at least one of the following properties hold:

\begin{itemize}
\item the value can be directly connected to a local variable in the source code using debug symbols;
\item the value is a not a constant index of the pointer operand of a GetElementPtr instruction \footnote{\url{https://llvm.org/doxygen/classllvm_1_1GetElementPtrInst.html}};
\item the value is related to a Load or Store instruction (both pointers and values);
\item the value is the return value of the function;
\end{itemize}

During the compilation, the pass dumps the information about program points and variables, such as type, comparability, and bounds, in a JSON file for each module. Then, each file is processed and merged to produce the declaration file for \daikon.

When executed, the dumper binary outputs a dtrace file related to the execution.

\begin{example}

Consider this example function in C language:

\begin{lstlisting}[language=C]
int isgreater(int x, int y) {

  if (x > y)
    return 1;
  return 0;

}
\end{lstlisting}

Converted to \llvm IR using \clang it is:

\begin{lstlisting}[language=llvm]
define dso_local i32 @isgreater(i32 %x, i32 %y)
    local_unnamed_addr #0 !dbg !16 {
entry:
  call void @llvm.dbg.value(metadata i32 %x, metadata !20,
                            metadata !DIExpression()), !dbg !22
  call void @llvm.dbg.value(metadata i32 %y, metadata !21,
                            metadata !DIExpression()), !dbg !22
  %cmp = icmp sgt i32 %x, %y, !dbg !23
  %. = zext i1 %cmp to i32, !dbg !22
  ret i32 %., !dbg !25
}
\end{lstlisting}

{\color{black} % for unknown reason pdflatex here changes also the font color

The \lstinline[columns=fixed, language=llvm]{llvm.dbg.value} intrinsic calls for the values \%x and \%y tell that these values are related to some variables in the source code, x and y.

In this small example, the instrumented values are \%x and \%y, because they are related to the source code variables, and \%., because it is the return value.

The program points in which each variable is logged are basic blocks and not functions. This is trivial to achieve just generating the decl and the dtrace files like if each basic block is a separate function, logging two times the variables at the end of the block to comply with the \daikon formats that need an \verb|ENTER| and at least one \verb|EXIT|.

}

\end{example}

%\daniele{il pezzo sopra del logging two times spiegalo con maggiori dettagli, non chiarissimo}

\subsection{Online learning}

We patched \daikon version 5.8.3 to enable online learning with the dumper binary. Logging all the dtrace files produced by the dumper binary for each input can be expensive in terms of disk space. Instead, we run the dumper inside \daikon and process the outcome on the fly.

After the learning, the invariants are saved as textual output in order to be processed in the next stage.

% todo discuss new command line

\subsection{Checks generation}

%\daniele{Per renderlo more smooth dovresti dire upfront che inserisci delle C functions per controllare ad ogni step se un inviarante holds etc}
The file with the invariants that are generated by \daikon is then parsed in this stage.

This file is textual, and for each program point lists the relations, that are expressions over one or more variables.

A simple one can be just \verb|LOC_x > 1| but more complex invariants are possible, such as \verb|LOC_x^2 + 3 * LOC_y - LOC_z >= 0|.

We parse this textual representation using regular expressions to locale the variable names and substitute them with C variables when generating C functions for each check.

The generated function is named \verb|__daikon_constr_ID| where ID is a unique number identifying the invariant.

The return value is the ID shifted by 1 if the invariant is violated, 0 otherwise.

\begin{example}

A generated function looks like the following C snippet:

\begin{lstlisting}[language=C]
uint32_t __daikon_constr_123(int32_t v0) {

  if (!(v0 > 1))
    return 123 << 1;
  return 0;

}
\end{lstlisting}

In this case, the ID is 123.

\end{example}

The script that generates the code creates also a JSON file describing each generated function for the next stage.

\subsection{Target compilation}

The final binary, ready to be fuzzed by \aflpp, is compiled using an \llvm pass that takes into account the outcome of the checks generation phase.

The \llvm values involved in the invariants are retrieved using their name that is the same between the dumper pass and the present one.

In this pass, the \verb|prev_loc| variable of the \aflpp instrumentation that tracks the incoming block when logging an edge is XOR-ed with the return value of each check function. When the invariant is not violated, the function returns 0, which results in normal edge coverage.

At the end of each block, the inserted instrumentation looks like the one in the following pseudocode:

\begin{lstlisting}[language=C]
__afl_area_ptr[cur_loc ^ prev_loc]++;
prev_loc = cur_loc >> 1;
prev_loc ^= __daikon_constr_123(variable1);
prev_loc ^= __daikon_constr_321(variable2, variable3);
...
\end{lstlisting}

Whenever an invariant is shared between two blocks and one dominates the other, the dominated block does not call again the function that performs the check, but directly reuses the outcome of the previously called function in the dominator.

All the instrumentation code inserted by the pass is marked with the \verb|nosanitize| metadata to avoid to be instrumented by sanitizers.

\begin{example}

At the end of the pipeline, the \lstinline[columns=fixed, language=C]{isgreater} function used in the previous examples, is compiled to \llvm IR as follows:

\begin{lstlisting}[language=llvm]
define dso_local i32 @isgreater(i32 %x, i32 %y)
    local_unnamed_addr #0 !dbg !7 {
entry:
  : AFL++ instrumentation
  %0 = load i32, i32* @__afl_prev_loc, !dbg !14, !nosanitize !2
  %1 = load i8*, i8** @__afl_area_ptr, !dbg !14, !nosanitize !2
  %2 = xor i32 %0, 2620, !dbg !14
  %3 = getelementptr i8, i8* %1, i32 %2, !dbg !14
  %4 = load i8, i8* %3, !dbg !14, !nosanitize !2
  %5 = add i8 %4, 1, !dbg !14
  %6 = icmp eq i8 %5, 0, !dbg !14
  %7 = zext i1 %6 to i8, !dbg !14
  %8 = add i8 %5, %7, !dbg !14
  store i8 %8, i8* %3, !dbg !14, !nosanitize !2
  store i32 1310, i32* @__afl_prev_loc, !dbg !14
  ; actual isgreater code
  call void @llvm.dbg.value(metadata i32 %x, metadata !12, metadata !DIExpression()), !dbg !14
  call void @llvm.dbg.value(metadata i32 %y, metadata !13, metadata !DIExpression()), !dbg !14
  %cmp = icmp sgt i32 %x, %y, !dbg !15
  %. = zext i1 %cmp to i32, !dbg !14
  ; check invariant and update prev_loc
  %9 = call i32 @__daikon_constr_1(i32 %.), !dbg !17
  %10 = load i32, i32* @__afl_prev_loc, !dbg !17, !nosanitize !2
  %11 = xor i32 %10, %9, !dbg !17
  store i32 %11, i32* @__afl_prev_loc, !dbg !17, !nosanitize !2
  ret i32 %., !dbg !17
}
\end{lstlisting}

There is a single invariant, checked by \lstinline[columns=fixed, language=C]{__daikon_constr_1}, and the return value of this function is XOR-ed with \verb|prev_loc|.

\end{example}

%\section{Discussion}

\chapter{Evaluation}
\label{cap:eval}

\iffalse
In this chapter, we present the results of a preliminary experimental investigation of our prototype \invscov. We choose as metrics: {\em checked invariants} to evaluate the effectiveness of the pruning stages, {\em triaged bugs} in a given time window to evaluate the ability to catch faults, and {\em executions per second} to evaluate the overhead.
\fi

In this chapter, we present the results of a preliminary experimental investigation of our prototype \invscov. We choose as metrics the heuristically {\em triaged bugs} in a given time window to evaluate the ability to catch faults, and {\em executions per second} to evaluate the overhead.

\section{Setup and Dataset}

All the experiments in this chapter were run on an x86\_64 machine with the Intel(R) Xeon(R) Platinum 8160 CPU at 2.10GHz and 32 GiB of RAM. The operating system used was Ubuntu 18.04 with the kernel version 4.15.

We selected a set of real-world applications to evaluate our prototype \invscov. Table \ref{tab:programs} lists the chosen programs. In the first part there are real-world programs, often old versions, that may contain bugs with a high probability. Most of them were part of evaluations in previous works from the literature (e.g.~\cite{vuzzer}~\cite{ankou}). In the second part, there are targets with known vulnerabilities taken from Fuzzer Test Suite~\cite{fts}, which are old versions of real-world programs too.

We report in Table \ref{tab:programs} the list of the programs, their versions, and the sanitizers used during the compilation of the subjects. The missing usage of UBSan on some programs is due to shallow bugs in those subjects that make them crash even with simple valid inputs. In this case, we opted to remove the sanitizer to allow the usage of the program in our evaluation.

In Table \ref{tab:cmd} we list the command line parameters used to run the programs. The subjects from Fuzzer Test Suite are not in this table because all of them use a \libfuzzer harness from OSS-Fuzz.

As described in Sec. \ref{sec:corpus}, the choice of the initial corpus for \invscov not trivial. 

\begin{table}[H]
\centering
\begin{tabular}{ |c|c|c| } 
 \hline
 Program & Version  & Sanitizers \\ 
 \hline
% \hline
 jasper & 2.0.16 & ASan, UBSan \\ 
 autotrace & 4333e37d5040881b19c2e1dad221f8e988419932 & ASan, UBSan \\ 
 cflow & 8a75c3721fd38f8d278cd71fb3682ead1497bb46 & ASan, UBSan \\ 
 
 catppt (catdoc) & 0.95 & ASan, UBSan \\ 
 xls2csv (catdoc) & 0.95 & ASan, UBSan \\ 
 
 pdf2cairo (poppler) & 53368f1717e88e40fe65d27e919c9abca11beac3 & ASan, UBSan \\ 
 potrace & 1.16  & ASan, UBSan \\ 
 pspp & 53d339111a9f51561cfccc65764874cdf54e501a & ASan \\
 exiv2 & 356f8627371e10cb8719eba3c45789e67420b10a & ASan, UBSan \\ 
 
 sndfile (libsndfile) & 2ccb23fe724d1d946b4e0c51b791cc655da6962e & ASan, UBSan \\ 
 
 \hline
 
 lcms & f9d75ccef0b54c9f4167d95088d4727985133c52 & ASan, UBSan \\
 re2 & 499ef7eff7455ce9c9fae86111d4a77b6ac335de & ASan, UBSan \\
 boringssl & 894a47df2423f0d2b6be57e6d90f2bea88213382 & ASan, UBSan \\
 
 guetzli & 9afd0bbb7db0bd3a50226845f0f6c36f14933b6b & ASan, UBSan \\
 libxml2 & v2.9.2 & ASan, UBSan \\
 woff2 & 9476664fd6931ea6ec532c94b816d8fbbe3aed90 & ASan \\
 
 libarchive & 51d7afd3644fdad725dd8faa7606b864fd125f88 & ASan \\
 pcre2 & 183 & ASan, UBSan \\
 
 \hline
\end{tabular}
\caption{Target programs versions and sanitizers.}
\label{tab:programs}
\end{table}

\begin{table}[H]
\centering
\begin{tabular}{ |c|c| } 
 \hline
 Program & Command line \\ 
 \hline
% \hline
 jasper & -f @@ -t jp2 -T mif -F /dev/null \\ 
 cflow & --no-main @@ \\
 catppt & @@ \\
 xls2csv & @@ \\
 pdf2cairo & -tiff @@ out \\
 potrace & -b pdf -c -q @@ -o /dev/null \\
 pspp & -O format=txt -o /dev/null -b @@ \\
 exiv2 & - (OSS-Fuzz harness) \\
 sndfile & --cart --instrument --broadcast @@ \\
 \hline
\end{tabular}
\caption{Command line used to run the target programs.}
\label{tab:cmd}
\end{table}

We opted to simply set 12 hours as the time window and run \afl to produce the corpus. This choice, of course, is not optimal in general but allows us to evaluate the technique on a large set of programs, reducing the manual work needed to observe when a fuzzer saturates for every single target. We leave the evaluation with incremental fuzzing after saturation with \invscov to future work.

All the benchmarks in this chapter are run with a time window of 48 hours and each experiment was repeated 3 times. The reported numbers are the median values.

% \section{Effectiveness of reducing likely invariants}

\section{Efficiency in finding faults}
\label{sec:bugs}

\begin{table}[h]
\centering
\begin{tabular}{ |c|c|c|c| }
 \hline
 Program & \aflpp bugs & \aflpp \invscov bugs & Intersection \\ 
 \hline
% \hline
 jasper     & 27 & 28 & 20 \\
 autotrace  & 30 & 28 & 28 \\ 
 cflow      & 6 & 6 & 4 \\ 
 catppt     & 4 & 8 & 4 \\ 
 xls2csv    & 19 & 29 & 18 \\ 
 pdf2cairo  & 25 & 27 & 21 \\ 
 potrace    & 1 & 0 & 0 \\ 
 pspp       & 34 & 17 & 12 \\ 
 exiv2      & 45 & 50 & 34 \\ 
 sndfile    & 18 & 21 & 18 \\ 
 \hline
 lcms       & 0 & 1 & 0 \\ 
 re2        & 1 & 0 & 0 \\ 
 boringssl  & 6 & 6 & 6 \\ 
 guetzli    & 3 & 1 & 1 \\ 
 libxml2    & 16 & 17 & 13 \\ 
 woff2      & 2 & 3 & 2 \\ 
 libarchive & 0 & 0 & 0 \\ 
 pcre2      & 104 & 122 & 57 \\ 
 \hline
 Total & 341 & 364 & 238 \\
 \hline
\end{tabular}
\caption{Triaged bugs found during the 48h experiments.}
\label{tbl:bugs}
\end{table}

In this section, we evaluate the ability of \invscov to find more or different bugs than the baseline \aflpp in a given time window of 48 hours.

As the number of reported crashes by the two fuzzers is high, we opted to automatically triage the crashes to the number of bugs heuristically using the hash of the call stack \footnote{The hash represents the sequence of functions that are concurrently active on the run-time stack at a given moment~\cite{callcontext}.} registered when the program crashes.
While this is a less sound metric than counting ground-truth bugs (e.g.~\cite{magma}), it was used in recent past works like~\cite{ankou} successfully.

To be more sound, we removed the addresses belonging to the C and the C++ standard libraries from the call stacks to reduce false positives.

In Table \ref{tbl:bugs} we report the number of uncovered bugs for \aflpp and \aflpp with \invscov. We report the intersection between the sets too, as a fuzzer that finds less but different bugs than another are still interesting.

\section{Performance overhead}

\begin{figure}[h]
  \includegraphics[scale=0.85]{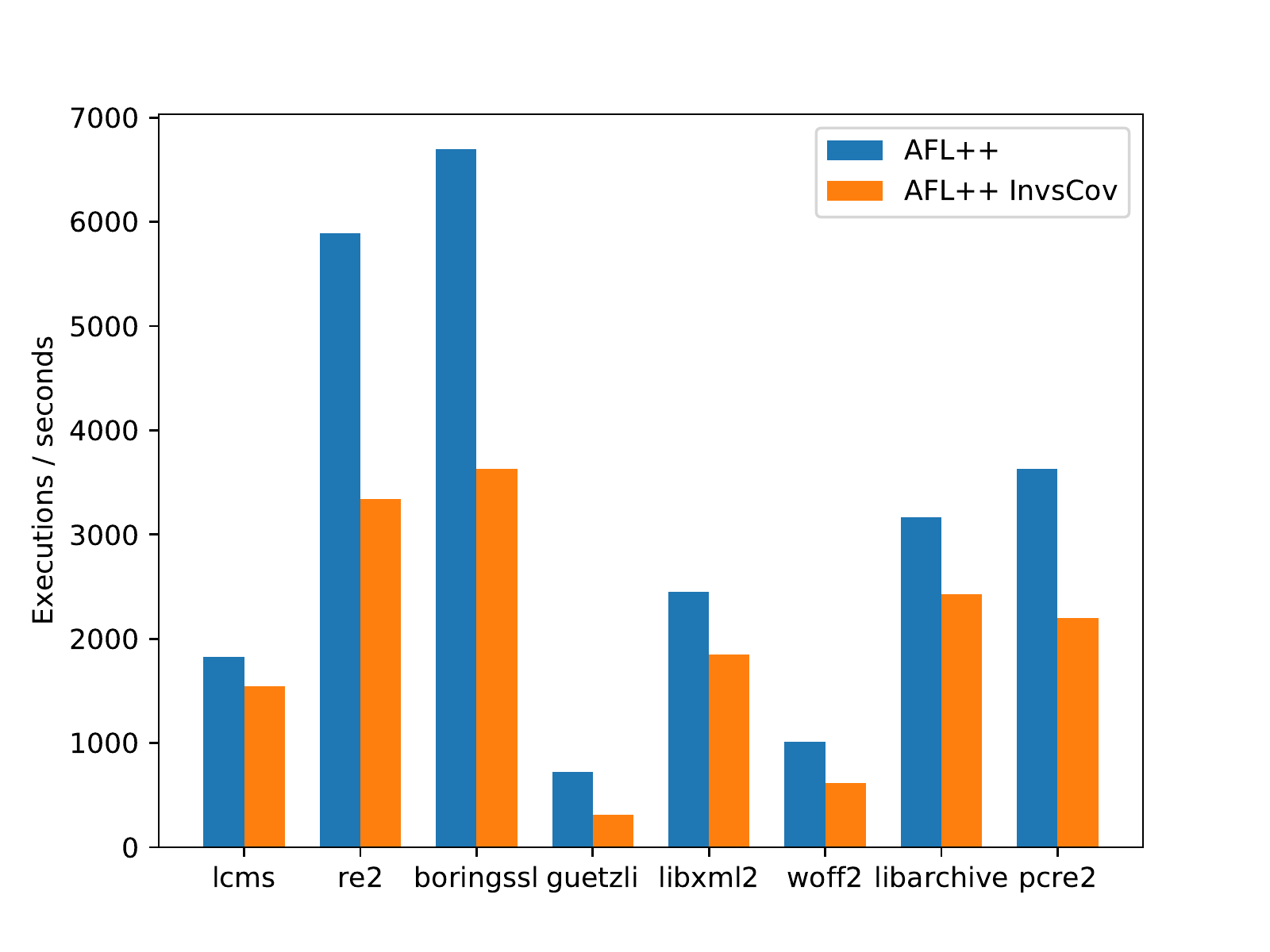}
  \centering
  \caption{Comparison in medium speed over a 48h experiment.}
  \label{fig:speedchart}
\end{figure}

In this section, we evaluate the performance overhead in terms of speed.

We compare the medium speed registered by the fuzzer (executions per second) for an entire 48-hour run. The chosen subjects are the targets used in Sec. \ref{sec:bugs} that come from Fuzzer Test Suite. These targets use the harnesses from OSS-Fuzz and so \aflpp can execute them in Persistent Mode, allowing a more sound comparison in terms of speed without the overhead of \verb|fork()| operations.

In Figure \ref{fig:speedchart} the medium executions per second of \aflpp and \aflpp \invscov are compared. The gap is almost never more than 2x, even smaller for targets with medium speed (2000 execs/sec).

The medium overhead over the set of benchmarks is 1.62x.

\section{Discussion}

From the evaluation of our technique for the ability to trigger faults, it appears that our technique performs well when the coverage produced by the initial corpus already covers most of the program. In targets like potrace, our technique fails to uncover a bug that is in a code region not covered by the initial corpus. On other targets, such as pcre2, \invscov performs well and uncovers more and different faults than vanilla \aflpp.

This highlights that the limitation in speed showed in the second part of the evaluation, can sometimes decrease the performance in uncovering bugs in new program points. This can be considered a limitation, or not because the purpose of the technique is not to be a replacement for Coverage-guided Fuzzing, but an incremental step instead. The problem can be addressed with a wiser choice of the initial corpus, taking the testcases of an already saturated fuzzer such as \aflpp itself.

The overall results suggest that \invscov improves the state-of-the-art of Fuzz Testing in terms of findings faults especially when the code is already explored by the initial corpus.

\chapter{Conclusion}
\label{cap:conclusion}

In this thesis, we introduced a new feedback for Feedback-driven Fuzz Testing in order to approximate the program state coverage better than traditional Coverage-guided Fuzzing.

%\daniele{data belongs to an?}
Reaching a program point does not guarantee the discovery of a fault in such a portion of code: we proposed to distinguish the same program point in the registered coverage if the values in the program state are unusual using likely invariants.

The likely invariants are mined local constraints that, if violated, may uncover the presence of a bug or a local unusual state. In both cases, they define a meaningful division of the possible values used in the corresponding program point.

In our technique, we learn invariants at the granularity of basic blocks to define a feedback that combines edge coverage and invariants violations.

In this way, we were able to augment the sensitivity of the coverage feedback of Feedback-driven Guided Fuzzing taking into account data and not only code without the typical path explosion issue that characterizes coverage types based on data tracking.

The developed prototype, \invscov, extends the compilation pipeline of \llvm to produce binaries that can be used to learn invariants and to record invariants and code coverage for \aflpp.

We showed that the prototype works on a set of real-world benchmarks producing fully functional binaries that can be easily fuzzed uncovering more of different types of faults than vanilla \aflpp with a reasonable performance overhead.

Based on our results, in general, we can devise that augmenting the sensitivity of a feedback, automatically and not just manually for a small set of program points and variables, with a sane amount of useful information can improve the search algorithm of the fuzzer.

\section{Future directions}

We foresee two directions of improvement: one is technical, the other concerns the implementation.

Firstly, our methodology can grow to learn the invariants on-demand during the run of the fuzzer. Instead of being a preprocessing step before the run of the fuzzer, we can adapt it to modify the target program on-the-fly with discovered likely invariants. In this way, we can use our technique on newly discovered code too, and not only on program points that are already covered by the execution of the input corpus.

%\daniele{a rough/coarse abstraction? o non ho capito il senso; fixa inglese dell'ultima frase il that simply etc}
Another avenue for technical improvement is the tracking of the side effect in the memory state. We employ such abstraction to avoid logging each value in the memory state, which would be an impossible solution in practice. In the future, we can augment such a method to not only track the values in the basic block state, which contains the side effects in memory, but also other interesting values that are implicitly related to the program state but not directly used in the block. For instance, consider the size of a buffer in the heap that is not used as a variable in a basic block that simply performs access to such buffer. The buffer size is still an interesting value to track because it has an indirect relationship with the memory address used in the block.

On the implementation side, we should replace the \daikon invariant detector with a more performant engine that uses the GPU. \daikon is a mono thread program in Java that runs on the CPU, while its algorithm can be reimplemented to exploit the parallelization opportunities from modern machines.

%\daniele{JIT compiler di che tipo? Hai JIT anche per managed runtimes tipo Java. Forse intendi JIT compiler of a dynamic binary translation system?}
Another future direction is the implementation of \invscov on the IR of an emulator. The \llvm implementation rarely uses information about the source code and so our technique, that does not depend on the particular frontend language, can be easily reimplemented as an instrumentation pass of a JIT compiler of a dynamic binary translator like \qemu.

% bibliography
\cleardoublepage
\phantomsection
\bibliographystyle{sapthesis} % BibTeX style
\bibliography{bibliography} % BibTeX database without .bib extension

\begin{thebibliography}{100}

\bibitem{gcc_bb}
{{Basic Blocks (GNU Compiler Collection (GCC) Internals)}}.
\newblock [Online; accessed 10-July-2020].
\newblock Available from:
  \url{https://gcc.gnu.org/onlinedocs/gccint/Basic-Blocks.html}.

\bibitem{testing_standard}
{IEEE Guide for Software Verification and Validation Plans}.
\newblock \emph{IEEE Std 1059-1993},  (1994), 1.

\bibitem{laf}
{Circumventing Fuzzing Roadblocks with Compiler Transformations}.
\newblock
  \url{https://lafintel.wordpress.com/2016/08/15/circumventing-fuzzing-roadblocks-with-compiler-transformations/}
  (2016).
\newblock [Online; accessed 10-Sep-2020].

\bibitem{glossary_standard}
{ISO/IEC/IEEE International Standard - Systems and software engineering --
  Vocabulary}.
\newblock \emph{C/S2ESC - Software \& Systems Engineering Standards Committee},
   (2017).

\bibitem{oss-fuzz}
{Google OSS-Fuzz: continuous fuzzing of open source software}.
\newblock \url{https://github.com/google/oss-fuzz} (2019).

\bibitem{cfg}
\textsc{Allen, F.~E.}
\newblock Control flow analysis.
\newblock \emph{SIGPLAN Not.}, \textbf{5} (1970), 1.
\newblock Available from: \url{https://doi.org/10.1145/390013.808479}, \href
  {http://dx.doi.org/10.1145/390013.808479} {\path{doi:10.1145/390013.808479}}.

\bibitem{nautilus}
\textsc{Aschermann, C., Frassetto, T., Holz, T., Jauernig, P., Sadeghi, A., and
  Teuchert, D.}
\newblock Nautilus: Fishing for deep bugs with grammars.
\newblock In \emph{NDSS} (2019).

\bibitem{whatthefuzz}
\textsc{Aschermann, C. and Schumilo, S.}
\newblock {What the fuzz} (2019).
\newblock [Online; accessed 10. Sep. 2020].
\newblock Available from: \url{https://hexgolems.com/talks/blackhat_19.pdf}.

\bibitem{blog_cornelius}
\textsc{Aschermann, C. and Schumilo, S.}
\newblock {On Measuring and Visualizing Fuzzer Performance} (2020).
\newblock [Online; accessed 1. Sep. 2020].
\newblock Available from:
  \url{https://hexgolems.com/2020/08/on-measuring-and-visualizing-fuzzer-performance/}.

\bibitem{ijon}
\textsc{Aschermann, C., Schumilo, S., Abbasi, A., and Holz, T.}
\newblock Ijon: Exploring deep state spaces via fuzzing.
\newblock In \emph{IEEE Symposium on Security and Privacy (Oakland)} (2020).

\bibitem{redqueen}
\textsc{Aschermann, C., Schumilo, S., Blazytko, T., Gawlik, R., and Holz, T.}
\newblock {REDQUEEN:} fuzzing with input-to-state correspondence.
\newblock In \emph{26th Annual Network and Distributed System Security
  Symposium, {NDSS}} (2019).
\newblock Available from:
  \url{https://www.ndss-symposium.org/ndss-paper/redqueen-fuzzing-with-input-to-state-correspondence/}.

\bibitem{aumasson2017automated}
\textsc{Aumasson, J.-P. and Romailler, Y.}
\newblock Automated testing of crypto software using differential fuzzing.
\newblock \emph{Black Hat USA, Jul},  (2017).

\bibitem{baldoni}
\textsc{Baldoni, R., Coppa, E., D'Elia, D.~C., Demetrescu, C., and Finocchi,
  I.}
\newblock A survey of symbolic execution techniques.
\newblock \emph{ACM Computing Surveys}, \textbf{51} (2018), 50:1.
\newblock Available from: \url{http://doi.acm.org/10.1145/3182657}, \href
  {http://dx.doi.org/10.1145/3182657} {\path{doi:10.1145/3182657}}.

\bibitem{qemu}
\textsc{Bellard, F.}
\newblock Qemu, a fast and portable dynamic translator.
\newblock In \emph{Proceedings of the Annual Conference on USENIX Annual
  Technical Conference}, ATEC '05, pp. 41--41. USENIX Association, Berkeley,
  CA, USA (2005).
\newblock Available from:
  \url{http://dl.acm.org/citation.cfm?id=1247360.1247401}.

\bibitem{blanchet2002introduction}
\textsc{Blanchet, B.}
\newblock Introduction to abstract interpretation.
\newblock  (2002).

\bibitem{grimoire}
\textsc{Blazytko, T., Aschermann, C., Schl{\"o}gel, M., Abbasi, A., Schumilo,
  S., W{\"o}rner, S., and Holz, T.}
\newblock {GRIMOIRE}: Synthesizing structure while fuzzing.
\newblock In \emph{28th {USENIX} Security Symposium ({USENIX} Security 19)},
  pp. 1985--2002. {USENIX} Association, Santa Clara, CA (2019).
\newblock ISBN 978-1-939133-06-9.
\newblock Available from:
  \url{https://www.usenix.org/conference/usenixsecurity19/presentation/blazytko}.

\bibitem{empiricalLaw}
\textsc{B{\"o}hme, M. and Falk, B.}
\newblock Fuzzing: On the exponential cost of vulnerability discovery.
\newblock In \emph{Proceedings of the 14th Joint meeting of the European
  Software Engineering Conference and the ACM SIGSOFT Symposium on the
  Foundations of Software Engineering}, ESEC/FSE, pp. 1--12 (2020).

\bibitem{entropic}
\textsc{B{\"o}hme, M., Man{\`e}s, V., and Cha, S.~K.}
\newblock Boosting fuzzer efficiency: An information theoretic perspective.
\newblock In \emph{Proceedings of the 14th Joint meeting of the European
  Software Engineering Conference and the ACM SIGSOFT Symposium on the
  Foundations of Software Engineering}, ESEC/FSE, pp. 1--11 (2020).

\bibitem{testing_efficiency}
\textsc{{B{\"o}hme}, M. and {Paul}, S.}
\newblock A probabilistic analysis of the efficiency of automated software
  testing.
\newblock \emph{IEEE Transactions on Software Engineering}, \textbf{42} (2016),
  345.

\bibitem{aflfast}
\textsc{B\"{o}hme, M., Pham, V.-T., and Roychoudhury, A.}
\newblock Coverage-based greybox fuzzing as markov chain.
\newblock In \emph{Proceedings of the 2016 ACM SIGSAC Conference on Computer
  and Communications Security}, CCS '16, pp. 1032--1043. Association for
  Computing Machinery, New York, NY, USA (2016).
\newblock ISBN 9781450341394.
\newblock Available from: \url{https://doi.org/10.1145/2976749.2978428}, \href
  {http://dx.doi.org/10.1145/2976749.2978428}
  {\path{doi:10.1145/2976749.2978428}}.

\bibitem{quickcheck}
\textsc{Claessen, K. and Hughes, J.}
\newblock Quickcheck: A lightweight tool for random testing of haskell
  programs.
\newblock In \emph{Proceedings of the Fifth ACM SIGPLAN International
  Conference on Functional Programming}, ICFP '00, pp. 268--279. Association
  for Computing Machinery, New York, NY, USA (2000).
\newblock ISBN 1581132026.
\newblock Available from: \url{https://doi.org/10.1145/351240.351266}, \href
  {http://dx.doi.org/10.1145/351240.351266} {\path{doi:10.1145/351240.351266}}.

\bibitem{halucinator}
\textsc{Clements, A.~A., Gustafson, E., Scharnowski, T., Grosen, P., Fritz, D.,
  Kruegel, C., Vigna, G., Bagchi, S., and Payer, M.}
\newblock Halucinator: Firmware re-hosting through abstraction layer emulation.
\newblock In \emph{29th {USENIX} Security Symposium ({USENIX} Security 20)},
  pp. 1201--1218. {USENIX} Association (2020).
\newblock ISBN 978-1-939133-17-5.
\newblock Available from:
  \url{https://www.usenix.org/conference/usenixsecurity20/presentation/clements}.

\bibitem{davide_php}
\textsc{Cova, M., Balzarotti, D., Felmetsger, V., and Vigna, G.}
\newblock Swaddler: An approach for the anomaly-based detection of state
  violations in web applications.
\newblock In \emph{Proceedings of the 10th International Symposium on Recent
  Advances in Intrusion Detection (RAID)}, pp. 63--86. Queensland, Australia
  (2007).

\bibitem{csallner2005check}
\textsc{Csallner, C. and Smaragdakis, Y.}
\newblock Check'n'crash: combining static checking and testing.
\newblock In \emph{Proceedings of the 27th international conference on Software
  engineering}, pp. 422--431 (2005).

\bibitem{ssa_convert}
\textsc{Cytron, R., Ferrante, J., Rosen, B.~K., Wegman, M.~N., and Zadeck,
  F.~K.}
\newblock Efficiently computing static single assignment form and the control
  dependence graph.
\newblock \emph{ACM Trans. Program. Lang. Syst.}, \textbf{13} (1991), 451.
\newblock Available from: \url{https://doi.org/10.1145/115372.115320}, \href
  {http://dx.doi.org/10.1145/115372.115320} {\path{doi:10.1145/115372.115320}}.

\bibitem{merge_vsa}
\textsc{Cytron, R., Ferrante, J., Rosen, B.~K., Wegman, M.~N., and Zadeck,
  F.~K.}
\newblock Efficiently computing static single assignment form and the control
  dependence graph.
\newblock \emph{ACM Trans. Program. Lang. Syst.}, \textbf{13} (1991), 451.
\newblock Available from: \url{https://doi.org/10.1145/115372.115320}, \href
  {http://dx.doi.org/10.1145/115372.115320} {\path{doi:10.1145/115372.115320}}.

\bibitem{daniele_pldi18}
\textsc{D'Elia, D.~C. and Demetrescu, C.}
\newblock On-stack replacement, distilled.
\newblock \emph{SIGPLAN Not.}, \textbf{53} (2018), 166.
\newblock Available from: \url{https://doi.org/10.1145/3296979.3192396}, \href
  {http://dx.doi.org/10.1145/3296979.3192396}
  {\path{doi:10.1145/3296979.3192396}}.

\bibitem{callcontext}
\textsc{D'Elia, D.~C., Demetrescu, C., and Finocchi, I.}
\newblock Mining hot calling contexts in small space.
\newblock \emph{Software: Practice and Experience}, \textbf{46} (2016), 1131.
\newblock Available from:
  \url{https://onlinelibrary.wiley.com/doi/abs/10.1002/spe.2348}, \href
  {http://arxiv.org/abs/https://onlinelibrary.wiley.com/doi/pdf/10.1002/spe.2348}
  {\path{arXiv:https://onlinelibrary.wiley.com/doi/pdf/10.1002/spe.2348}},
  \href {http://dx.doi.org/10.1002/spe.2348} {\path{doi:10.1002/spe.2348}}.

\bibitem{retrowrite}
\textsc{Dinesh, S., Burow, N., Xu, D., and Payer, M.}
\newblock Retrowrite: Statically instrumenting cots binaries for fuzzing and
  sanitization.
\newblock In \emph{IEEE S\&P 2020} (2020).

\bibitem{peach}
\textsc{Eddington, M.}
\newblock {Peach fuzzing platform}.
\newblock
  \url{https://web.archive.org/web/20180621074520/http://community.peachfuzzer.com/WhatIsPeach.html}.
\newblock [Online; accessed 10-Sep-2020].

\bibitem{ernst2001dynamically}
\textsc{Ernst, M.~D., Cockrell, J., Griswold, W.~G., and Notkin, D.}
\newblock Dynamically discovering likely program invariants to support program
  evolution.
\newblock \emph{IEEE Transactions on Software Engineering}, \textbf{27} (2001),
  99.

\bibitem{ernst_phd}
\textsc{Ernst, M.~D. and Notkin, D.}
\newblock \emph{Dynamically Discovering Likely Program Invariants}.
\newblock Ph.D. thesis, USA (2000).
\newblock AAI9983472.

\bibitem{daikon}
\textsc{Ernst, M.~D., Perkins, J.~H., Guo, P.~J., McCamant, S., Pacheco, C.,
  Tschantz, M.~S., and Xiao, C.}
\newblock The {Daikon} system for dynamic detection of likely invariants.
\newblock \emph{Science of Computer Programming}, \textbf{69} (2007), 35.

\bibitem{gamozo}
\textsc{Falk, B.}
\newblock {Brandon Falk @gamozolabs FuzzBench feedback} (2020).
\newblock [Online; accessed 1. Sep. 2020].
\newblock Available from: \url{https://github.com/google/fuzzbench/issues/654}.

\bibitem{weizz}
\textsc{Fioraldi, A., D'Elia, D.~C., and Coppa, E.}
\newblock {WEIZZ}: Automatic grey-box fuzzing for structured binary formats.
\newblock In \emph{Proceedings of the 29th ACM SIGSOFT International Symposium
  on Software Testing and Analysis}, ISSTA 2020. Association for Computing
  Machinery, New York, NY, USA (2020).
\newblock ISBN 9781450380089.
\newblock Available from: \url{https://doi.org/10.1145/3395363.3397372}, \href
  {http://dx.doi.org/10.1145/3395363.3397372}
  {\path{doi:10.1145/3395363.3397372}}.

\bibitem{qasan}
\textsc{Fioraldi, A., D'Elia, D.~C., and Querzoni, L.}
\newblock Fuzzing binaries for memory safety errors with {QASan}.
\newblock In \emph{2020 IEEE Secure Development Conference (SecDev)} (2020).

\bibitem{aflplusplus}
\textsc{Fioraldi, A., Maier, D., Ei{\ss}feldt, H., and Heuse, M.}
\newblock {AFL++}: Combining incremental steps of fuzzing research.
\newblock In \emph{14th {USENIX} Workshop on Offensive Technologies ({WOOT}
  20)}. {USENIX} Association (2020).

\bibitem{floyd1993assigning}
\textsc{Floyd, R.~W.}
\newblock Assigning meanings to programs.
\newblock In \emph{Program Verification}, pp. 65--81. Springer (1993).

\bibitem{revngfuzz}
\textsc{Frighetto, A.}
\newblock Fuzzing binaries with llvm's libfuzzer and rev.ng (2020).
\newblock [Online; accessed 20. Sep. 2020].
\newblock Available from: \url{https://rev.ng/blog/page-1.html}.

\bibitem{giuffrida_psi}
\textsc{Giuffrida, C., Cavallaro, L., and Tanenbaum, A.}
\newblock Practical automated vulnerability monitoring using program state
  invariants.
\newblock In \emph{Proceedings of the 43rd Annual IEEE/IFIP International
  Conference on Dependable Systems and Networks}, pp. 1--12. IEEE CS (2013).
\newblock ISBN 978-1-4673-6471-3.
\newblock \href {http://dx.doi.org/10.1109/DSN.2013.6575318}
  {\path{doi:10.1109/DSN.2013.6575318}}.

\bibitem{godefroid2007random}
\textsc{Godefroid, P.}
\newblock Random testing for security: blackbox vs. whitebox fuzzing.
\newblock In \emph{Proceedings of the 2nd international workshop on Random
  testing: co-located with the 22nd IEEE/ACM International Conference on
  Automated Software Engineering (ASE 2007)}, pp. 1--1 (2007).

\bibitem{sage}
\textsc{Godefroid, P., Levin, M.~Y., and Molnar, D.}
\newblock Automated whitebox fuzz testing (2008).
\newblock Available from:
  \url{https://www.microsoft.com/en-us/research/publication/automated-whitebox-fuzz-testing/}.

\bibitem{learnfuzz}
\textsc{Godefroid, P., Peleg, H., and Singh, R.}
\newblock Learn\&fuzz: Machine learning for input fuzzing.
\newblock In \emph{Proceedings of the 32Nd IEEE/ACM International Conference on
  Automated Software Engineering}, ASE 2017, pp. 50--59. IEEE Press,
  Piscataway, NJ, USA (2017).
\newblock ISBN 978-1-5386-2684-9.
\newblock Available from:
  \url{http://dl.acm.org/citation.cfm?id=3155562.3155573}.

\bibitem{fts}
\textsc{Google}.
\newblock Fuzzer test suite.
\newblock [Online; accessed 1. Sep. 2020].
\newblock Available from: \url{https://github.com/google/fuzzer-test-suite}.

\bibitem{saturation}
\textsc{Groce, A. and Regehr, J.}
\newblock {The Saturation Effect in Fuzzing}.
\newblock \url{https://blog.regehr.org/archives/1796}.

\bibitem{diduce}
\textsc{Hangal, S. and Lam, M.~S.}
\newblock Tracking down software bugs using automatic anomaly detection.
\newblock In \emph{Proceedings of the 24th International Conference on Software
  Engineering}, ICSE '02, pp. 291--301. Association for Computing Machinery,
  New York, NY, USA (2002).
\newblock ISBN 158113472X.
\newblock Available from: \url{https://doi.org/10.1145/581339.581377}, \href
  {http://dx.doi.org/10.1145/581339.581377} {\path{doi:10.1145/581339.581377}}.

\bibitem{range_analysis}
\textsc{Harrison, W.}
\newblock Compiler analysis of the value ranges for variables.
\newblock \emph{IEEE Transactions on Software Engineering}, \textbf{3} (1977),
  243.
\newblock \href {http://dx.doi.org/10.1109/TSE.1977.231133}
  {\path{doi:10.1109/TSE.1977.231133}}.

\bibitem{magma}
\textsc{Hazimeh, A., Herrera, A., and Payer, M.}
\newblock Magma: A ground-truth fuzzing benchmark.
\newblock  (2020).
\newblock \href {http://arxiv.org/abs/2009.01120} {\path{arXiv:2009.01120}}.

\bibitem{hoare72}
\textsc{Hoare, C.~A.}
\newblock Proof of correctness of data representations.
\newblock \emph{Acta Inf.}, \textbf{1} (1972), 271.
\newblock Available from: \url{https://doi.org/10.1007/BF00289507}, \href
  {http://dx.doi.org/10.1007/BF00289507} {\path{doi:10.1007/BF00289507}}.

\bibitem{Hoare1969}
\textsc{Hoare, C. A.~R.}
\newblock {An Axiomatic Basis for Computer Programming}.
\newblock \emph{Comm.~ACM}, \textbf{12} (1969), 576.

\bibitem{langfuzz}
\textsc{Holler, C., Herzig, K., and Zeller, A.}
\newblock Fuzzing with code fragments.
\newblock In \emph{21st {USENIX} Security Symposium ({USENIX} Security 12)},
  pp. 445--458. {USENIX} Association, Bellevue, WA (2012).
\newblock ISBN 978-931971-95-9.
\newblock Available from:
  \url{https://www.usenix.org/conference/usenixsecurity12/technical-sessions/presentation/holler}.

\bibitem{pangolin}
\textsc{Huang, H., Yao, P., Wu, R., Shi, Q., and Zhang, C.}
\newblock Pangolin: Incremental hybrid fuzzing with polyhedral path
  abstraction.
\newblock In \emph{2020 IEEE Symposium on Security and Privacy (SP)}, pp.
  1613--1627. IEEE Computer Society, Los Alamitos, CA, USA (2020).
\newblock Available from:
  \url{https://doi.ieeecomputersociety.org/10.1109/SP40000.2020.00063}, \href
  {http://dx.doi.org/10.1109/SP40000.2020.00063}
  {\path{doi:10.1109/SP40000.2020.00063}}.

\bibitem{fuzzbench}
\textsc{Jonathan~Metzman, L.~S., Abhishek~Arya}.
\newblock {FuzzBench}: Fuzzer benchmarking as a service.
\newblock Google Security Blog (2020).

\bibitem{fuzzeval}
\textsc{Klees, G., Ruef, A., Cooper, B., Wei, S., and Hicks, M.}
\newblock Evaluating fuzz testing.
\newblock In \emph{Proceedings of the 2018 ACM SIGSAC Conference on Computer
  and Communications Security}, CCS '18, pp. 2123--2138. Association for
  Computing Machinery, New York, NY, USA (2018).
\newblock ISBN 9781450356930.
\newblock Available from: \url{https://doi.org/10.1145/3243734.3243804}, \href
  {http://dx.doi.org/10.1145/3243734.3243804}
  {\path{doi:10.1145/3243734.3243804}}.

\bibitem{lamport1977proving}
\textsc{Lamport, L.}
\newblock Proving the correctness of multiprocess programs.
\newblock \emph{IEEE transactions on software engineering},  (1977), 125.

\bibitem{llvm}
\textsc{Lattner, C.}
\newblock \emph{{LLVM: An Infrastructure for Multi-Stage Optimization}}.
\newblock Master's thesis, {Computer Science Dept., University of Illinois at
  Urbana-Champaign}, Urbana, IL (2002).
\newblock {\em See {\tt http://llvm.cs.uiuc.edu}.}

\bibitem{LattnerAdve:tutorial}
\textsc{Lattner, C. and Adve, V.}
\newblock {The LLVM Compiler Framework and Infrastructure Tutorial}.
\newblock In \emph{{LCPC'04 Mini Workshop on Compiler Research
  Infrastructures}}. West Lafayette, Indiana (2004).

\bibitem{laycock1993theory}
\textsc{Laycock, G.~T.}
\newblock \emph{The theory and practice of specification based software
  testing}.
\newblock Ph.D. thesis, Citeseer.

\bibitem{ariane}
\textsc{Le~Lann, G.}
\newblock An analysis of the ariane 5 flight 501 failure - a system engineering
  perspective.
\newblock In \emph{Proceedings of the 1997 International Conference on
  Engineering of Computer-Based Systems}, ECBS'97, pp. 339--346. IEEE Computer
  Society, USA (1997).
\newblock ISBN 0818678895.

\bibitem{perffuzz}
\textsc{Lemieux, C., Padhye, R., Sen, K., and Song, D.}
\newblock Perffuzz: Automatically generating pathological inputs.
\newblock In \emph{Proceedings of the 27th ACM SIGSOFT International Symposium
  on Software Testing and Analysis}, ISSTA 2018, pp. 254--265. Association for
  Computing Machinery, New York, NY, USA (2018).
\newblock ISBN 9781450356992.
\newblock Available from: \url{https://doi.org/10.1145/3213846.3213874}, \href
  {http://dx.doi.org/10.1145/3213846.3213874}
  {\path{doi:10.1145/3213846.3213874}}.

\bibitem{therac}
\textsc{Leveson, N.~G. and Turner, C.~S.}
\newblock An investigation of the therac-25 accidents.
\newblock \emph{Computer}, \textbf{26} (1993), 18.
\newblock Available from: \url{https://doi.org/10.1109/MC.1993.274940}, \href
  {http://dx.doi.org/10.1109/MC.1993.274940}
  {\path{doi:10.1109/MC.1993.274940}}.

\bibitem{value-profile}
\textsc{{{LLVM Project}}}.
\newblock {LibFuzzer - Value Profile}.
\newblock Available from:
  \url{{https://llvm.org/docs/LibFuzzer.html#value-profile}}.

\bibitem{libfuzzer}
\textsc{{{LLVM Project}}}.
\newblock {libFuzzer {\textendash} a library for coverage-guided fuzz testing.}
  (2018).
\newblock Available from: \url{https://llvm.org/docs/LibFuzzer.html}.

\bibitem{mopt}
\textsc{Lyu, C., Ji, S., Zhang, C., Li, Y., Lee, W.-H., Song, Y., and Beyah,
  R.}
\newblock {MOPT}: Optimized mutation scheduling for fuzzers.
\newblock In \emph{28th {USENIX} Security Symposium ({USENIX} Security 19)},
  pp. 1949--1966. {USENIX} Association, Santa Clara, CA (2019).
\newblock ISBN 978-1-939133-06-9.
\newblock Available from:
  \url{https://www.usenix.org/conference/usenixsecurity19/presentation/lyu}.

\bibitem{unicorefuzz}
\textsc{Maier, D., Radtke, B., and Harren, B.}
\newblock Unicorefuzz: On the viability of emulation for kernelspace fuzzing.
\newblock In \emph{13th {USENIX} Workshop on Offensive Technologies ({WOOT}
  19)}. {USENIX} Association, Santa Clara, CA (2019).
\newblock Available from:
  \url{https://www.usenix.org/conference/woot19/presentation/maier}.

\bibitem{survey}
\textsc{Man{\`{e}}s, V. J.~M., Han, H., Han, C., Cha, S.~K., Egele, M.,
  Schwartz, E.~J., and Woo, M.}
\newblock The art, science, and engineering of fuzzing: A survey.
\newblock \emph{{IEEE} Transactions on Software Engineering}, \textbf{xxx}
  (2019), xxx.

\bibitem{ankou}
\textsc{Man{\`{e}}s, V. J.~M., Kim, S., and Cha, S.~K.}
\newblock Ankou: Guiding grey-box fuzzing towards combinatorial difference.
\newblock pp. 1024--1036 (2020).

\bibitem{zeller_issta20}
\textsc{Mathis, B., Gopinath, R., and Zeller, A.}
\newblock Learning input tokens for effective fuzzing.
\newblock In \emph{{ISSTA} '20: 29th {ACM} {SIGSOFT} International Symposium on
  Software Testing and Analysis, Virtual Event, USA, July 18-22, 2020} (edited
  by S.~Khurshid and C.~S. Pasareanu), pp. 27--37. {ACM} (2020).
\newblock Available from: \url{https://doi.org/10.1145/3395363.3397348}, \href
  {http://dx.doi.org/10.1145/3395363.3397348}
  {\path{doi:10.1145/3395363.3397348}}.

\bibitem{softw_book}
\textsc{{Maur\'{i}cio Aniche}}.
\newblock {Software Testing: From Theory to Practice} (2020).
\newblock Available from: \url{https://sttp.site/}.

\bibitem{marius}
\textsc{{M}uench, M., {S}tijohann, J., {K}argl, F., {F}rancillon, A., and
  {B}alzarotti, D.}
\newblock {W}hat you corrupt is not what you crash: {C}hallenges in fuzzing
  embedded devices.
\newblock In \emph{{NDSS} 2018, {N}etwork and {D}istributed {S}ystems
  {S}ecurity {S}ymposium, 18-21 {F}ebruary 2018, {S}an {D}iego, {CA}, {USA}}.
  {S}an {D}iego, {UNITED} {STATES} (2018).
\newblock Available from: \url{http://www.eurecom.fr/publication/5417}.

\bibitem{valgrind}
\textsc{Nethercote, N. and Seward, J.}
\newblock Valgrind: A framework for heavyweight dynamic binary instrumentation.
\newblock In \emph{Proceedings of ACM SIGPLAN 2007 Conference on Programming
  Language Design and Implementation (PLDI 2007)}, pp. 89--100. San Diego,
  California, USA (2007).

\bibitem{unicornemu}
\textsc{Ngyuen, A.~Q. and Dang, H.~V.}
\newblock Unicorn: Next generation cpu emulator framework (2020).
\newblock Available from:
  \url{http://www.unicorn-engine.org/BHUSA2015-unicorn.pdf}.

\bibitem{parmesan}
\textsc{\"Osterlund, S., Razavi, K., Bos, H., and Giuffrida, C.}
\newblock {ParmeSan}: {Sanitizer}-guided {Greybox} {Fuzzing}.
\newblock In \emph{{USENIX} {Security}} (2020).
\newblock Available from:
  \url{Paper=https://download.vusec.net/papers/parmesan_sec20.pdf
  Code=https://github.com/vusec/parmesan}.

\bibitem{zest}
\textsc{Padhye, R., Lemieux, C., Sen, K., Papadakis, M., and Le~Traon, Y.}
\newblock Semantic fuzzing with zest.
\newblock In \emph{Proceedings of the 28th ACM SIGSOFT International Symposium
  on Software Testing and Analysis}, ISSTA 2019, pp. 329--340. Association for
  Computing Machinery, New York, NY, USA (2019).
\newblock ISBN 9781450362245.
\newblock Available from: \url{https://doi.org/10.1145/3293882.3330576}, \href
  {http://dx.doi.org/10.1145/3293882.3330576}
  {\path{doi:10.1145/3293882.3330576}}.

\bibitem{fuzzfactory}
\textsc{Padhye, R., Lemieux, C., Sen, K., Simon, L., and Vijayakumar, H.}
\newblock Fuzzfactory: Domain-specific fuzzing with waypoints.
\newblock \emph{Proc. ACM Program. Lang.}, \textbf{3} (2019).
\newblock Available from: \url{https://doi.org/10.1145/3360600}, \href
  {http://dx.doi.org/10.1145/3360600} {\path{doi:10.1145/3360600}}.

\bibitem{pattabiraman2010automated}
\textsc{Pattabiraman, K., Saggese, G.~P., Chen, D., Kalbarczyk, Z., and Iyer,
  R.}
\newblock Automated derivation of application-specific error detectors using
  dynamic analysis.
\newblock \emph{IEEE Transactions on Dependable and Secure Computing},
  \textbf{8} (2010), 640.

\bibitem{usbfuzz}
\textsc{Peng, H. and Payer, M.}
\newblock Usbfuzz: A framework for fuzzing {USB} drivers by device emulation.
\newblock In \emph{29th {USENIX} Security Symposium ({USENIX} Security 20)},
  pp. 2559--2575. {USENIX} Association (2020).
\newblock ISBN 978-1-939133-17-5.
\newblock Available from:
  \url{https://www.usenix.org/conference/usenixsecurity20/presentation/peng}.

\bibitem{tfuzz}
\textsc{{Peng}, H., {Shoshitaishvili}, Y., and {Payer}, M.}
\newblock T-fuzz: Fuzzing by program transformation.
\newblock In \emph{2018 IEEE Symposium on Security and Privacy (SP)}, pp.
  697--710 (2018).
\newblock Available from: \url{https://doi.org/10.1109/SP.2018.00056}, \href
  {http://dx.doi.org/10.1109/SP.2018.00056} {\path{doi:10.1109/SP.2018.00056}}.

\bibitem{aflsmart}
\textsc{{Pham}, V., {Boehme}, M., {Santosa}, A.~E., {Caciulescu}, A.~R., and
  {Roychoudhury}, A.}
\newblock Smart greybox fuzzing.
\newblock \emph{IEEE Transactions on Software Engineering},  (2019).
\newblock \href {http://dx.doi.org/10.1109/TSE.2019.2941681}
  {\path{doi:10.1109/TSE.2019.2941681}}.

\bibitem{aflnet}
\textsc{Pham, V., B{\"o}hme, M., and Roychoudhury, A.}
\newblock Aflnet: A greybox fuzzer for network protocols.
\newblock In \emph{Proceedings of the 13rd IEEE International Conference on
  Software Testing, Verification and Validation : Testing Tools Track} (2020).

\bibitem{sebastian}
\textsc{Poeplau, S. and Francillon, A.}
\newblock Symbolic execution with symcc: Don{\textquoteright}t interpret,
  compile!
\newblock In \emph{29th {USENIX} Security Symposium ({USENIX} Security 20)},
  pp. 181--198. {USENIX} Association (2020).
\newblock ISBN 978-1-939133-17-5.
\newblock Available from:
  \url{https://www.usenix.org/conference/usenixsecurity20/presentation/poeplau}.

\bibitem{dominator}
\textsc{Prosser, R.~T.}
\newblock Applications of boolean matrices to the analysis of flow diagrams.
\newblock In \emph{Papers Presented at the December 1-3, 1959, Eastern Joint
  IRE-AIEE-ACM Computer Conference}, IRE-AIEE-ACM '59 (Eastern), pp. 133--138.
  Association for Computing Machinery, New York, NY, USA (1959).
\newblock ISBN 9781450378680.
\newblock Available from: \url{https://doi.org/10.1145/1460299.1460314}, \href
  {http://dx.doi.org/10.1145/1460299.1460314}
  {\path{doi:10.1145/1460299.1460314}}.

\bibitem{integer_range}
\textsc{Quintao~Pereira, F.~M., Rodrigues, R.~E., and Sperle~Campos, V.~H.}
\newblock A fast and low-overhead technique to secure programs against integer
  overflows.
\newblock In \emph{Proceedings of the 2013 IEEE/ACM International Symposium on
  Code Generation and Optimization (CGO)}, CGO '13, pp. 1--11. IEEE Computer
  Society, USA (2013).
\newblock ISBN 9781467355247.
\newblock Available from: \url{https://doi.org/10.1109/CGO.2013.6494996}, \href
  {http://dx.doi.org/10.1109/CGO.2013.6494996}
  {\path{doi:10.1109/CGO.2013.6494996}}.

\bibitem{vuzzer}
\textsc{Rawat, S., Jain, V., Kumar, A., Cojocar, L., Giuffrida, C., and Bos,
  H.}
\newblock Vuzzer: Application-aware evolutionary fuzzing.
\newblock In \emph{24th Annual Network and Distributed System Security
  Symposium, {NDSS}} (2017).
\newblock Available from:
  \url{https://www.ndss-symposium.org/ndss2017/ndss-2017-programme/vuzzer-application-aware-evolutionary-fuzzing/}.

\bibitem{reed1991purify}
\textsc{Reed~Hastings, B.~J.}
\newblock Purify: Fast detection of memory leaks and access errors.
\newblock In \emph{In proc. of the winter 1992 usenix conference}. Citeseer
  (1991).

\bibitem{rice}
\textsc{Rice, H.~G.}
\newblock Classes of recursively enumerable sets and their decision problems.
\newblock \emph{Transactions of the American Mathematical Society}, \textbf{74}
  (1953), 358.
\newblock Available from: \url{http://www.jstor.org/stable/1990888}.

\bibitem{ssa}
\textsc{Rosen, B.~K., Wegman, M.~N., and Zadeck, F.~K.}
\newblock Global value numbers and redundant computations.
\newblock In \emph{Proceedings of the 15th ACM SIGPLAN-SIGACT Symposium on
  Principles of Programming Languages}, POPL '88, pp. 12--27. Association for
  Computing Machinery, New York, NY, USA (1988).
\newblock ISBN 0897912527.
\newblock Available from: \url{https://doi.org/10.1145/73560.73562}, \href
  {http://dx.doi.org/10.1145/73560.73562} {\path{doi:10.1145/73560.73562}}.

\bibitem{kafl}
\textsc{Schumilo, S., Aschermann, C., Gawlik, R., Schinzel, S., and Holz, T.}
\newblock Kafl: Hardware-assisted feedback fuzzing for os kernels.
\newblock In \emph{Proceedings of the 26th USENIX Conference on Security
  Symposium}, SEC'17, pp. 167--182. USENIX Association, USA (2017).
\newblock ISBN 9781931971409.

\bibitem{asan}
\textsc{Serebryany, K., Bruening, D., Potapenko, A., and Vyukov, D.}
\newblock Addresssanitizer: A fast address sanity checker.
\newblock In \emph{Proceedings of the 2012 USENIX Conference on Annual
  Technical Conference}, USENIX ATC'12, p.~28. USENIX Association (2012).

\bibitem{driller}
\textsc{Stephens, N., Grosen, J., Salls, C., Dutcher, A., Wang, R., Corbetta,
  J., Shoshitaishvili, Y., Kruegel, C., and Vigna, G.}
\newblock Driller: Augmenting fuzzing through selective symbolic execution.
\newblock In \emph{NDSS}, vol.~16, pp. 1--16 (2016).

\bibitem{honggfuzz}
\textsc{Swiecki, R.}
\newblock {Honggfuzz}.
\newblock [Online; accessed 1. Sep. 2020].
\newblock Available from: \url{https://github.com/google/honggfuzz}.

\bibitem{tillmann2006discovering}
\textsc{Tillmann, N., Chen, F., and Schulte, W.}
\newblock Discovering likely method specifications.
\newblock In \emph{International Conference on Formal Engineering Methods}, pp.
  717--736. Springer (2006).

\bibitem{vfuzz}
\textsc{Vranken, G.}
\newblock {VrankenFuzz a multi-sensor, multi-generator mutational fuzz testing
  engine}.
\newblock
  \url{https://guidovranken.files.wordpress.com/2018/07/vrankenfuzz.pdf}
  (2018).

\bibitem{syzkaller}
\textsc{Vyukov, D.}
\newblock syzkaller - kernel fuzzer.
\newblock [Online; accessed 10. Sep. 2020].
\newblock Available from: \url{https://github.com/google/syzkaller}.

\bibitem{becollab}
\textsc{Wang, J., Duan, Y., Song, W., Yin, H., and Song, C.}
\newblock Be sensitive and collaborative: Analyzing impact of coverage metrics
  in greybox fuzzing.
\newblock In \emph{22nd International Symposium on Research in Attacks,
  Intrusions and Defenses ({RAID} 2019)}, pp. 1--15. {USENIX} Association,
  Chaoyang District, Beijing (2019).
\newblock ISBN 978-1-939133-07-6.
\newblock Available from:
  \url{https://www.usenix.org/conference/raid2019/presentation/wang}.

\bibitem{error_part}
\textsc{{Weyuker}, E.~J. and {Jeng}, B.}
\newblock Analyzing partition testing strategies.
\newblock \emph{IEEE Transactions on Software Engineering}, \textbf{17} (1991),
  703.

\bibitem{csmith}
\textsc{Yang, X., Chen, Y., Eide, E., and Regehr, J.}
\newblock Finding and understanding bugs in c compilers.
\newblock In \emph{Proceedings of the 32nd ACM SIGPLAN Conference on
  Programming Language Design and Implementation}, PLDI '11, p. 283–294.
  Association for Computing Machinery, New York, NY, USA (2011).
\newblock ISBN 9781450306638.
\newblock Available from: \url{https://doi.org/10.1145/1993498.1993532}, \href
  {http://dx.doi.org/10.1145/1993498.1993532}
  {\path{doi:10.1145/1993498.1993532}}.

\bibitem{qsym}
\textsc{Yun, I., Lee, S., Xu, M., Jang, Y., and Kim, T.}
\newblock Qsym: A practical concolic execution engine tailored for hybrid
  fuzzing.
\newblock In \emph{Proceedings of the 27th USENIX Conference on Security
  Symposium}, SEC'18, pp. 745--761. USENIX Association, USA (2018).
\newblock ISBN 9781931971461.

\bibitem{aflwhitepaper}
\textsc{Zalewski, M.}
\newblock {American Fuzzy Lop - Whitepaper}.
\newblock \url{https://lcamtuf.coredump.cx/afl/technical_details.txt} (2016).

\bibitem{firmafl}
\textsc{Zheng, Y., Davanian, A., Yin, H., Song, C., Zhu, H., and Sun, L.}
\newblock Firm-afl: High-throughput greybox fuzzing of iot firmware via
  augmented process emulation.
\newblock In \emph{28th {USENIX} Security Symposium ({USENIX} Security 19)},
  pp. 1099--1114. {USENIX} Association, Santa Clara, CA (2019).
\newblock ISBN 978-1-939133-06-9.
\newblock Available from:
  \url{https://www.usenix.org/conference/usenixsecurity19/presentation/zheng}.

\end{thebibliography}

\end{document}